\documentclass[a4paper,USenglish]{lipics-v2016}
\usepackage{xspace}
\usepackage[all,british]{foreign} 
\usepackage{epsfig}
\usepackage{epstopdf}
\usepackage{graphicx}
\usepackage{todonotes}
\usepackage{amsmath,amssymb,amsthm}
\usepackage{xcolor}
\usepackage{proof}
\usepackage{latexsym}
\usepackage{amssymb}
\usepackage{color}
\usepackage{comment}
\usepackage{thm-restate}
\usepackage{psfrag}
\usepackage{comment}
\usepackage[noend]{algorithmic}
\usepackage{algorithm}

\usepackage{microtype}

\newcommand{\strongcontract}{{\sc Path-contraction Preserving Strong Connectivity}}
\newcommand{\strongnodedel}{{\sc Vertex-deletion Preserving Strong Connectivity}}

\newcommand{\twoconndel}{{\sc Weighted Biconnectivity Deletion}}

\newtheorem*{theorem*}{Theorem}
\newtheorem*{lemma*}{Lemma}
\theoremstyle{plain}
\newtheorem{claim}[theorem]{Claim}

\newtheorem{observation}[theorem]{Observation}

\newtheorem{redr}[theorem]{Reduction Rule}

\renewcommand{\geq}{\geqslant}
\renewcommand{\leq}{\leqslant}

\newcommand{\reals}{{\mathbb R}}
\newcommand{\naturals}{{\mathbb N}}

\newcommand{\lv}[1]{}
\newcommand{\sv}[1]{}

\newcommand{\contract}{\mathbin{/\mkern-6mu/}}

\newcommand{\heavy}{{\sf Heavy}}
\newcommand{\segment}{{\sf Segment}}
\newcommand{\component}{{\sf Component}}

\newcommand{\critical}{{\sf Critical}}
\newcommand{\partner}{{\sf Partner}}

\newcommand{\cP}{{\mathcal{P}}}
\newcommand{\cH}{{\mathcal{H}}}
\newcommand{\old}[1]{}

\newcommand{\cF}{{\mathcal{F}}}
\newcommand{\cB}{{\mathcal{B}}}

\newenvironment{tightcenter}
 {\parskip=0pt\par\nopagebreak\centering}
 {\par\noindent\ignorespacesafterend}

\usepackage{tikz}
\usetikzlibrary{calc}
\usepackage{xargs}
\usepackage{xifthen}
\usepackage{framed}

\usepackage{ctable}

\newlength{\RoundedBoxWidth}
\newsavebox{\GrayRoundedBox}
\newenvironment{GrayBox}[1]%
   {\setlength{\RoundedBoxWidth}{\textwidth-4.5ex}
    \def\boxheading{#1}
    \begin{lrbox}{\GrayRoundedBox}
       \begin{minipage}{\RoundedBoxWidth}%
   }{%
       \end{minipage}
    \end{lrbox}%
    \begin{tightcenter}%
    \begin{tikzpicture}%
       \node(Text)[draw=black!20,fill=white,rounded corners,%
             inner sep=2ex,text width=\RoundedBoxWidth]%
             {\usebox{\GrayRoundedBox}};
        \coordinate(x) at (current bounding box.north west);
        \node [draw=white,rectangle,inner sep=3pt,anchor=north west,fill=white] 
        at ($(x)+(6pt,.75em)$) {\boxheading};
    \end{tikzpicture}
    \end{tightcenter}\vspace{0pt}%
    \ignorespacesafterend
}    

\newenvironment{problem}[2][]{\noindent\ignorespaces%
                                \FrameSep=6pt%
                                \parindent=0pt%
                \vspace*{-.5em}
                \ifthenelse{\isempty{#1}}{%
                  \begin{GrayBox}{\textsc{#2}}%
                }{%
                  \begin{GrayBox}{\textsc{#2} parameterized by~{#1}}%
                }
                \newcommand\Prob{Problem:}%
                \newcommand\Input{Input:}%
                \begin{tabular*}{\textwidth}{@{\hspace{.1em}} >{\itshape} p{1.2cm} p{0.85\textwidth} @{}}%
            }{
                \end{tabular*}%
                \end{GrayBox}%
                \vspace*{-.5em}
                \ignorespacesafterend
            } 

\title{Path-contractions, edge deletions and connectivity preservation\footnote{The research of Gregory Gutin  was partially supported by Royal Society Wolfson Research Merit Award. M. S. Ramanujan acknowledges support from Austrian Science Fund (FWF, project P26696).}}

\author[1]{Gregory Gutin}\author[2]{M. S. Ramanujan}\author[1]{ Felix Reidl}\author[1]{ Magnus Wahlstr\"{o}m}

 \affil[1]{Royal Holloway University, University of London, Egham, UK.     \texttt{\{G.Gutin|Magnus.Wahlstrom\}@rhul.ac.uk}, \texttt{felix.reidl@gmail.com}}      
 \affil[2]{Algorithms and Complexity Group, TU Wien, Vienna, Austria.     \texttt{ramanujan@ac.tuwien.ac.at}} 
  
\authorrunning{G. Gutin, M. S. Ramanujan, F. Reidl and M. Wahlstr\"{o}m} 

\Copyright{Gregory Gutin, M. S. Ramanujan, Felix Reidl and Magnus Wahlstr\"{o}m}


\EventEditors{John Q. Open and Joan R. Acces}
\EventNoEds{2}
\EventLongTitle{42nd Conference on Very Important Topics (CVIT 2016)}
\EventShortTitle{CVIT 2016}
\EventAcronym{CVIT}
\EventYear{2016}
\EventDate{December 24--27, 2016}
\EventLocation{Little Whinging, United Kingdom}
\EventLogo{}
\SeriesVolume{42}
\ArticleNo{23}

\begin{document}
\maketitle

\begin{abstract}
  We study several problems related to graph modification problems under
  connectivity constraints from the perspective of parameterized complexity:
  {\sc (Weighted) Biconnectivity Deletion}, where we are tasked with deleting~$k$
  edges while preserving biconnectivity in an undirected graph, {\sc Vertex-deletion
  Preserving Strong Connectivity}, where we want to maintain strong connectivity
  of a digraph while deleting exactly~$k$ vertices, and {\sc Path-contraction Preserving
  Strong Connectivity}, in which the operation of path contraction on arcs is used instead.
  The parameterized tractability of this last problem was 
  posed by Bang-Jensen and Yeo [DAM 2008] as an open question and we answer it
  here in the negative: both variants of preserving strong connectivity are
  $\sf W[1]$-hard. Preserving biconnectivity, on the other hand, turns out to be
  fixed parameter tractable and we provide a $2^{O(k\log k)} n^{O(1)}$-algorithm
  that solves {\sc Weighted Biconnectivity Deletion}. Further, we show 
  that the unweighted case even admits a randomized polynomial kernel. 
  All our results provide further interesting data points for the systematic
  study of connectivity-preservation constraints in the parameterized setting.
\end{abstract}

\section{Introduction}

Some of the most well studied classes of network design problems involve
starting with a given network and making modifications to it so that the
resulting network satisfies certain connectivity requirements, for instance a
prescribed edge- or vertex-connectivity. This class of problems has a long and
rich history (see \eg  \cite{bang2008digraphs,frank11}) and has recently started 
to be examined through the lens of parameterized complexity. Under this paradigm,
we ask whether a (hard) problem admits an algorithm with a running time $f(k)n^{O(1)}$, 
where $n$ is the size of the input, $k$ the {\em parameter}, and $f$ some computable function.
A natural parameter to consider in this context is the number of editing operations
allowed and we can reasonably assume that this number is small compared to the size of
the graph. 

To approach this line of research systematically, let us identify the `moving
parts' of the broader question of editing under connectivity-constraints:
first and foremost, the network in question might best be modelled as either a
directed or undirected graph, potentially with edge- or vertex-weights. This,
in turn, informs the type of connectivity we restrict, \eg strong connectivity
or fixed value of edge-/vertex-connectivity. Additionally, the connectivity
requirement might be non-uniform, \ie it might be specified for individual
vertex-pairs. The constraint one operates under might either be to
\emph{preserve}, to \emph{augment}, or to \emph{decrease} said connectivity.
Finally, we need to fix a suitable editing operation; besides the obvious
vertex- and edge-removal, more intricate operations like edge contractions are
possible.

While not all possible combinations of these factors might result in a problem
that currently has an immediate real-world application, they are nonetheless
important data points in the systematic study of algorithmic tractability. For
example, if we fix the editing operation to be the addition of edges (often
called `links' in this context) and our goal is to increase connectivity, then
the resulting class of {\em connectivity augmentation problems} has been
thoroughly researched. We refer to the monograph by Frank~\cite{frank11} for
further results on polynomial-time solvable cases and approximation
algorithms. Under the parameterized complexity paradigm, Nagamochi~\cite{Nagamochi200383} and Guo and
Uhlmann~\cite{GuoU10} studied the problem of augmenting a 1-edge-
connected graph with~$k$ links to a 2-edge-connected graph. Nagamochi obtained an {\sf FPT} algorithm for this problem while Guo and Uhlmann showed that
this problem, alongside its vertex-connectivity variant, admits a quadratic
kernel. Marx and V\'{e}gh~\cite{marx2015fixed} studied the more general
problem of augmenting the edge-connectivity of an undirected graph from
$\lambda - 1$ to $\lambda$, via a minimum set of links that has a total {\em
cost} of at most~$k$, and obtained an {\sf FPT} algorithm as well as a
polynomial kernel for this problem. 
Basavaraju \etal~\cite{basavaraju2014} improved the running time of
their algorithm and further showed the fixed-parameter tractability of a dual
parameterization of this problem.

A second large body of work can be found in the antithetical class of problems,
where we ask to {\em delete} edges from a network while {\em preserving} connectivity.
Probably the most studied member of these {\em connectivity preservation problems}
is the {\sc Minimum Strong Spanning Spanning Subgraph} (MSSS) problem: given a
strongly connected digraph we are asked to find a strongly connected subgraph
with a minimum number of arcs. The problem is NP-complete 
(an easy reduction from the {\sc Hamiltonian
Cycle} problem) and there exist a number of approximation algorithms for it (see
the monograph by Bang-Jensen and Gutin for details and references~\cite{bang2008digraphs}).
Bang-Jensen and Yeo \cite{bang2008minimum} were the
first to study {\sc MSSS} from the parameterized complexity perspective. They
presented an algorithm that runs in time $2^{O(k \log k)}n^{O(1)}$ and decides
whether a given strongly connected digraph $D$ on $n$ vertices and $m$ arcs
has a strongly connected subgraph with at most $m-k$ arcs provided $m \geq
2n-2$. Basavaraju \etal~\cite{basavarajuMRS17} extended this result not only
to arbitrary number~$m$ of arcs but also to $\lambda$-arc-strong connectivity
for an arbitrary integer~$\lambda$, and they further extended it to
$\lambda$-edge-connected {\em undirected} graphs.

We consider the undirected variant of this problem, however, we aim to
preserve the {\em vertex-connectivity} instead of edge-connectivity. As noted
by Marx and V\'{e}gh~\cite{marx2015fixed}, vertex-connectivity variants of
parameterized connectivity problems seem to be much harder to approach than
their edge-connectivity counterparts.\footnote{Marx and V\'{e}gh~\cite{marx2015fixed} compare 
\cite{WatanabeN87} and \cite{Frank92} to \cite{FrankJ95} and \cite{Vegh11} with respect to polynomial-time exact and approximation algorithms.}
Moreover,  even the complexity of the problem of
augmenting the vertex-connectivity of an undirected graph from~$2$ to~$3$, via
a minimum set of up to $k$ new links remains open~\cite{marx2015fixed}. Our
main result in this direction is the first {\sf FPT} algorithm for the following
problem\footnote{Note that since 1-vertex-connectivity is trivially
equivalent to 1-edge-connectivity, the $1$-vertex-connectivity case
was proved to be {\sf FPT} by Basavaraju \etal~\cite{basavarajuMRS17}.}:

\begin{problem}[$k$]{\twoconndel}
  \Input & A biconnected graph $G$,  $k\in \naturals$, 
           $w^*\in \reals_{\geq 0}$ and a function $w:E(G)\to \reals_{\geq 0}$. \\
  \Prob  & Is there a set $S \subseteq E(G)$ of size at most $k$ such that 
           $G - S$ is biconnected and $w(S) \geq w^*$? 
\end{problem}

\begin{restatable}{theorem}{thmbiconnected}\label{thm:biconnected}
  {\twoconndel} can be solved in time $2^{O(k \log k)}n^{O(1)}$.
\end{restatable}

\noindent 
We further show that this problem has a randomized polynomial kernelization when
the edges are required to have only unit weights. To be precise, all
inputs for the unweighted variant \textsc{Unweighted Biconnectivity Deletion}
are of the form $(G,k,w^*,w)$, where $w^*=k$ and $w(e)=1$ for every $e\in E(G)$.

\begin{restatable}{theorem}{kernelbiconnected}\label{thm:kernelbiconnected}
  {\sc Unw. Biconnectivity Del.} has a randomized
  kernel with $O(k^9)$ vertices. \looseness-1
\end{restatable}

\noindent 
Along with arc-additions and arc-deletions, a third interesting operation on
digraphs is the \emph{path-contraction} operation which has been used to
obtain structural results on paths in digraphs~\cite{bang2008digraphs}. To
{\em path-contract} an arc~$(x,y)$ in a digraph~$D$, we remove it from~$D$,
identify $x$ and~$y$ and keep the in-arcs of~$x$ and the our-arcs of~$y$ for
the combined vertex. The resulting digraph is denoted by $D \contract (x,y)$.
It is useful to extend this notation to sequences of contractions:
let $S=(a_1,a_2,\dots ,a_p)$ be a sequence of arcs of a digraph $D$. Then $D
\contract S$ is defined as $(\dots ((D \contract a_1) \contract a_2) \contract
\dots ) \contract a_p$. Since the resulting digraph does not depend on the order of
the arcs~\cite{bang2008digraphs}, this notation can equivalently
be used for arc-sets.

Bang-Jensen and Yeo~\cite{bang2008minimum} asked whether the problem of path-
contracting at least $k$ arcs to maintain strong connectivity of a given
digraph $D$ is fixed-parameter tractable. Formally, the problem is stated as
follows:

\begin{problem}[k]{\strongcontract}
  \Input & A strongly connected digraph $D$ and an integer $k$. \\
  \Prob  & Is there a sequence $S=(a_1,\dots,a_k)$ of arcs of $D$ 
           such that $D \contract S$ is also strongly connected?
\end{problem}

\noindent Our first result is a negative answer to the  question of Bang-Jensen and Yeo.
That is, we show that  this problem is unlikely to be {\sf FPT}. 

\begin{restatable}{theorem}{thmpathcontractmain}\label{thm:pathcontract}
	{\strongcontract} is {\sf W[1]}-hard.
\end{restatable}

\noindent 
We follow up this result by considering a natural vertex-deletion variant of
the problem and extending our {\sf W[1]}-hardness result to this problem as well. In
this variant, the objective is to check for the existence of a set of {\em
exactly} $k$ vertices such that on deleting these vertices from the given
digraph, the digraph stays strongly connected.

\begin{restatable}{theorem}{thmnodedelmain}\label{thm:nodedel}
	{\strongnodedel} is {\sf W[1]}-hard.
\end{restatable}

\medskip\noindent
{\bf Our Methodology.}\quad%
Our algorithm for {\twoconndel} builds upon the recent  approach introduced by
Basavaraju~\etal~\cite{basavarajuMRS17} to handle connectivity
preservation problems, in particular the $p$-{\sc $\lambda$-Edge Connected
Subgraph} ($p$-{\sc $\lambda$-ECS}) problem where the objective is to
delete $k$ edges while keeping the graph $\lambda$-edge connected. Call an
edge {\em deletable} (we refer to it as {\em non-critical} in the case of
vertex-connectivity) if deleting it keeps the given (di)graph
$\lambda$-edge connected, {\em undeletable} ({\em critical}) otherwise, and call an edge
\emph{irrelevant} if there is a solution disjoint from the edge.

For an even value of $\lambda$ and a $\lambda$-edge-connected undirected
graph $G$, Basavaraju~\etal~\cite{basavarajuMRS17} proved that unless the
total number of deletable edges is bounded by $O(\lambda k^2$), it is
possible in polynomial time to obtain a set $F$ of $k$ edges such that $G-F$ is still
$\lambda$-edge-connected. This result does not hold for odd values of
$\lambda$ as can be seen, \eg, when~$\lambda=1$ and~$G$ is a cycle. In this
much more involved case, unless the total number of deletable edges is bounded
by $O(\lambda k^3)$, it is possible in polynomial time to obtain either a set $F$ of
$k$ edges such that $G-F$ is still $\lambda$-edge-connected or to identify 
an irrelevant edge.

{\twoconndel} is similar to the case of odd~$\lambda$ as we find either a
solution or an irrelevant edge. The main difference between our {\sc FPT} algorithm
and the one presented by Basavaraju~\etal is the deep structural analysis necessitated
by the shift from edge-connectivity to vertex-connectivity:
While in the former case the failure to find a solution
means that~$G$ can be decomposed into a `cycle-like' structure, in our
case no such simple structure arises. Instead, we perform a careful 
examination of mixed cuts in the graph, each of which comprise precisely 
one critical edge~$e$ and a vertex $w$ which we call the {\em partner}
of~$e$. We show that either a large number of critical edges share a common partner or there is a
large number of critical edges with pairwise distinct partners. In the former
case, we proof the existence of an irrelevant edge while in the latter
case we are able to construct a solution. Our result is based on a non-trivial combination
of several new structural properties of biconnected graphs and critical
edges which we believe is of independent interest and
useful in the study of other connectivity-constrained problems.

The kernel stated in Theorem \ref{thm:kernelbiconnected} relies on the
powerful {\em cut-covering lemma}  of Kratsch and
Wahlstr\"om~\cite{KratschW12} which has been central to the development of
several recent kernelization algorithms~\cite{Kratsch14}. While Basavaraju
\etal obtained a randomized compression for the $p$-{\sc $\lambda$-ECS}
problem using sketching techniques from dynamic graph algorithms, we provide
an alternative approach and show that when dealing with biconnectivity it is
also possible to obtain a (randomized) polynomial {\em kernel}. We believe
that this approach could be applicable for higher values of vertex-
connectivity and for other connectivity deletion problems, as long as one is
able to bound the number of critical or undeletable edges in the given
instance by an appropriate function of the parameter.

\medskip\noindent
{\bf Further related work.} In the  \textsc{Minimum Equivalent Digraph}
problem, given a digraph $D$, the aim is to find a spanning subgraph $H$ of
$D$ with minimum number of arcs such that if there is an $x$-$y$ directed path
in $D$ then there is such a path in $H$ for every pair $x,y$ of vertices  of
$D$. Since it is not hard to solve \textsc{Minimum Equivalent Digraph} for
acyclic digraphs, \textsc{Minimum Equivalent Digraph} for general digraphs can
be reduced to MSSS in polynomial time.  Chapter 12 of the monograph of Bang-
Jensen and Gutin~\cite{bang2008digraphs} surveys pre-2009 results on
\textsc{Minimum Equivalent Digraph}.
The first exact algorithm for the  \textsc{Mnimum Equivalent Digraph} problem,
running in time $2^{O(m)}$,  was given by Moyles and
Thompson~\cite{moyles1969algorithm} in 1969, where $m$ is the number of arcs
in the graph. More recently, Fomin, Lokshtanov, and Saurabh~\cite{fomin2016efficient} gave the
first vertex-exponential algorithm for this problem, \ie an algorithm with a
running time of $2^{O(n)}$.

\section{Preliminaries}

\medskip\noindent
{\bf Graphs.}\quad%
For an undirected graph $G$ and vertex set $S\subseteq V(G)$, we denote by $E(S)$ the set of
edges of $G$ with both endpoints in $S$. For a pair of disjoint vertex sets
$X,Y\subseteq V(G)$, we denote by $E(X,Y)$ the set of edges with one endpoint
in $X$ and the other in $Y$. For a vertex set $X\subseteq V(G)$, we denote by
$N_G(X)$ the set of vertices of $V(G)\setminus X$ which are adjacent to a
vertex in $X$. We denote by $\delta_G(X)$ the set $E(X,V(G)\setminus X)$.
A vertex in a connected undirected graph is a {\em cut-vertex} if deleting
this vertex disconnects the graph. 
A {\em biconnected graph} is a connected graph on two or more vertices having no cut-vertices.

For a directed or undirected path $P$, we denote by $V(P)$ and $E(P)$ the 
set of vertices and edges in~$P$, respectively. We further
denote by $V_{\sf int}(P)$ the set of internal vertices of $P$.

We say that two paths $P_1$ and $P_2$ are {\em
internally vertex-disjoint} if $V_{\sf int}(P_1) \cap V_{\sf int}(P_2)
=\emptyset$. Note that under this definition, a path consisting of a single
vertex is internally vertex-disjoint to any other path.

For two internally vertex-disjoint paths $P_1=v_1,\dots v_t$ and
$P_2=w_1,\dots, w_q$ such that $v_1\neq w_1$ and $v_t=w_1$, we denote by
$P_1+P_2$ the concatenated path $v_1,\dots, v_{t-1},v_t,w_2,\dots, w_q$. When
we deal with undirected graphs, we will abuse this notation and also
use $P_1+P_2$ to refer to the concatenated path that arises when $v_1=w_1$ and
$v_t\neq w_q$ or $v_1=w_q$ and $w_1\neq v_t$ or $w_1=v_t$ and $v_1\neq w_q$.
In short, the two `orientations' of any undirected path are used
interchangeably and when we need to differentiate between the two
orientations, we explicitly say that we are {\em traversing} the path from one
specified endpoint to the other.

\begin{definition}
  Let $G$ be a graph and $x,y\in V(G)$ two vertices. An $x$-$y$ {\em separator} (an $x$-$y$
  {\em cut}) is a set $S\subseteq V(G)\setminus \{x,y\}$ (respectively
  $S\subseteq E(G)$) such that there is no $x$-$y$ path in $G-S$. A {\em mixed
  $x$-$y$ cut} is a set $S\subseteq V(G)\cup E(G)$ such that $|S\cap E(G)|=1$
  and there is no $x$-$y$ path in $G-S$. 
\end{definition}

Let $S\subseteq V(G)\cup E(G)$. We denote by $R_G(x,S)$ the set of vertices in
the same connected component as $x$ in the graph $G-S$.  The reference to
$G$ is dropped if it is clear from the context.

\begin{definition}
  Let $G$ be a graph and $x,y\in V(G)$. Let $\cP$ be a set of internally
  vertex-disjoint $x$-$y$ paths in $G$. Then, we call $\cP$ an $x$-$y$ {\em
  flow}. The {\em value} of this flow is $|\cP|$. We say that
  an edge $e$ {\em participates} in the $x$-$y$ flow $\cP$ if $e\in \bigcup_{P
  \in \cP} P$.
\end{definition}

We denote by $\kappa_G(x,y)$ the value of the maximum $x$-$y$ flow in $G$
with the reference to~$G$ dropped when clear from the context. 

Recall that Menger's theorem states that for distinct \emph{non-adjacent}
vertices~$x$ and~$y$, the size of the smallest $x$-$y$ separator is precisely
$\kappa(x,y)$. We extend the definition of flows to vertex sets as
follows. Let $x\in V(G)$ and $Y\subseteq V(G)$ be such that $x\notin Y$. Let
$\cP$ be a set of paths in $G$ which have an endpoint in $Y$ and intersect
only in $x$. Then, we refer to $\cP$ as an $x$-$Y$ flow, with the value of
this flow defined as $|\cP|$.

\medskip\noindent
{\bf Directed graphs.}\quad%
We will refer to edges in a digraph as {\em arcs}. For a vertex~$x$
in a digraph~$D$ we write~$N^-_D(x)$ and~$N^+_D(x)$ to denote its
in- and out-neighbours, respectively.
A {\em sink} is a vertex with no out-neighbours and a {\em source} is  a vertex
with no in-neighbours. While we will use path-contraction in digraphs only for
single arcs, \ie directed paths of length one, we restate the more general
definition for context.

\begin{definition}[Bang-Jensen and Gutin~\cite{bang2008digraphs}]
  Let $P$ be an $(x,y)$-path in a  directed multigraph $D$. Then, $D \contract
  P$ denotes the multigraph obtained from $D$ by deleting all vertices of $P$
  and adding a new vertex $z$ such that every arc with head $x$ (tail $y$) and
  tail (respectively head) in $V\setminus V(P)$ becomes an arc with head
  (tail) $z$ and the same tail (respectively head).
\end{definition}

\noindent
The path-contraction of a single arc~$(x,y)$ is equivalent to 
identifying the vertices~$x$ and~$y$ as a new vertex~$z$ and then removing
the resulting loop as well as all arcs from~$z$ to~$N^+(x)$ and~$N^-(y)$.

\medskip\noindent
{\bf Parameterized Complexity.}\quad%
An instance of a parameterized problem $\Pi$
is a pair $(I,k)$ where $I$ is the \emph{main part} and $k$ is the
\emph{parameter}; the latter is usually a non-negative integer.  
A parameterized problem is
\emph{fixed-parameter tractable} if there exists a computable function
$f$ such that instances $(I,k)$ can be solved in time $O(f(k)|{I}|^c)$
where $|I|$ denotes the size of~$I$. The class of all fixed-parameter
tractable decision problems is called {\sf FPT} and algorithms which run in
the time specified above are called {\sf FPT} algorithms.

To establish that a problem under a specific parameterization is not
in {\sf FPT} (under common complexity-theoretic assumptions) we provide
\emph{parameter-preserving reductions} from problems known to lie in intractable
classes like~$\sf W[1]$ or~$\sf W[2]$. In such a reduction, an instance~$(I_1,k_1)$
is reduced in polynomial time to an instance~$(I_2,k_2)$
where~$k_2 \leq f(k_1)$ for some function~$f$. In the context of this paper we will
use that {\sc Independent Set} under its natural parameterization (the size of the
independent set) is $\sf W[1]$-hard~\cite{CyganFKLMPPS15}.

A {\em reduction rule} for a parameterized problem $\Pi$  is an algorithm that
given an instance $(I,k)$ of a problem $\Pi$ returns an instance $(I',k')$ of
the \emph{same} problem. The reduction rule is said to be {\em sound} if it
holds that  $(I,k) \in \Pi$ if and only if $(I',k') \in \Pi$. A
\emph{kernelization} is a polynomial-time algorithm that given any instance
$(I,k)$ returns an instance $(I',k')$ such that $(I,k) \in \Pi$ if and only if
$(I',k') \in \Pi$ {\em and} $|I'|+k'\leq f(k)$ for some computable function
$f$. The function $f$ is called the {\em size} of the kernelization, and we
have a polynomial kernelization if $f(k)$ is polynomially bounded in $k$.  A
{\em randomized kernelization} is an algorithm which is allowed to err with
certain probability. That is, the returned instance will be equivalent to the
input instance only with a  certain probability.

\section{Preserving strong connectivity}

In this section, we prove Theorem~\ref{thm:pathcontract} and Theorem~\ref{thm:nodedel}.
%


\thmpathcontractmain*

\begin{proof}
   We reduce {\sc Independent Set} to {\strongcontract}.

	\medskip
	
	\noindent
	{\bf Construction.}   Let $(G,k)$ be an instance of {\sc Independent Set}. We now define a digraph $D$ as follows. We begin with the vertex set of $D$.	
	For every vertex $v\in V(G)$, $D$ has two vertices $v^-,v^+$.  For every edge $e=(u,v)\in E(G)$, the digraph $D$ has $k+2$ vertices $\hat e,\hat e_1,\dots, \hat e_{k+1}$. Finally, there are $2k+4$ special vertices $x,y,x^1,\dots,x^{k+1},y^1,\dots, y^{k+1}$. This completes the definition of $V(D)$. We now define the arc set of $D$ (see Figure~\ref{fig:path_contraction_reduction}).
	\begin{itemize}
\item For every $v\in V(G)$, we add the arc $(v^-,v^+)$ in $D$ .
\item 	For every $i\in [k+1]$,
	we add the arcs $\{(x,x^i),(x^i,x),(y,y^i),(y^i,y),(y,x)\}$. 
	\item For every edge $e=(u,v)\in E(G)$ and $i\in [k+1]$,  we add the arcs $\{(\hat e,\hat e_{i}),(\hat e_i,\hat e),$  $(v^-,\hat e),$ $(\hat e,v^+),$ $(u^-,\hat e),(\hat e,u^+)\}$ in $D$ .
	\item For every $v\in V(G)$, we add the arc $(x,v^-)$ and the arc $(v^+,y)$. 
\end{itemize}
This completes the construction of the digraph $D$. Clearly, $D$ is strongly-connected. 
	
	For an edge $e=(u,v)\in E(G)$, we denote by $\cB_e$ the set of arcs $\{(v^-,\hat e),(\hat e,v^+),$ $(u^-,\hat e),(\hat e,u^+)\}$ and by $\cF_e$, the set of arcs $\cB_e\cup \{(\hat e,\hat e_i),(\hat e_i,\hat e)|i\in [k+1]\}$  $\cup \{(u^-,u^+),(v^-,v^+),$ $(x,v^-),(v^+,y)$, $(x,u^-),(u^+,y),(y,x)\}$. We refer to the subgraph of $D$ induced by $\cF_e$ as the \emph{edge-selection gadget} in $D$ corresponding to $e$ (see Figure~\ref{fig:path_contraction_reduction}). The intuition here is that, as we will prove formally, any solution in $D$ will contain at most one of the two arcs $(u^-,u^+),(v^-,v^+)$.
	
	\medskip
	\noindent
	{\bf Proof of correctness.} We now argue that $(G,k)$ is a yes-instance of {\sc Independent Set} if and only if $(D,k)$ is a yes-instance of {\strongcontract}.
In the forward direction, suppose that $(G,k)$ is a yes-instance of {\sc Independent Set} and let $X\subseteq V(G)$ be a solution. Observe that $S=\{(v^-,v^+) \mid v\in X\}$ is a pairwise vertex-disjoint set of arcs. We claim that $S$ is a solution for the instance $(D,k)$. That is, $|S|\geq k$ and $D \contract S$ is strongly connected. The former is true by definition. We now argue the latter.
	
	\begin{claim}
	$D'=D \contract S$ is strongly connected.
		
	\end{claim}

	\begin{proof}
	Observe that it is sufficient to prove that $D''=D-Q$ is strongly connected, where $$Q=\bigcup_{v\in X}(N^+(v^-)\cup N^-(v^+))\setminus \{(v^-,v^+)\},$$ and $N^+(v) \mbox{ and } N^-(v) \mbox{ are the sets of out-neighbours and in-neighbours of $v$.}$	
	In other words,  for every arc $(v^-,v^+)\in S$, $Q$ contains all the arcs that are lost when we path-contract this arc. We begin by observing that the set $Q$ is disjoint from $\{(v^-,v^+) \mid v\in V(G)\}$. This follows from the definition of $Q$.
Due to this observation and the presence of the arc $(y,x)$, it follows that the vertices in $${\cal P}=\{(v^-,v^+) \mid v\in V(G)\}\cup  \{x,y,x^1,\dots,x^{k+1},y^1,\dots, y^{k+1}\}$$ occur in a single strongly connected component of $D''$. Hence, it suffices to argue that for every $e\in E(G)$, the vertex $\hat e$ is also in the same strongly connected component of $D''$.
	
	Note that since $X$ is an independent set in $G$, it must be the case that for any $e=(u,v)\in E(G)$, either $(u^-,u^+)\notin S$ or $(v^-,v^+)\notin S$. But this implies that either $\{(u^-,\hat e),(\hat e,u^+)\}\cap Q=\emptyset$ or $\{(v^-,\hat e),(\hat e,v^+)\}\cap Q=\emptyset$, implying that $\hat e$ is also in the same strongly connected component as the vertices in ${\cal P}$. Thus, $D''$ is strongly connected and so is $D'$.  This completes the proof of the claim and hence proves the correctness of  the forward direction of the reduction.
	\end{proof}
	
	\begin{figure}[t]
  \begin{center}
    \includegraphics[scale=0.5]{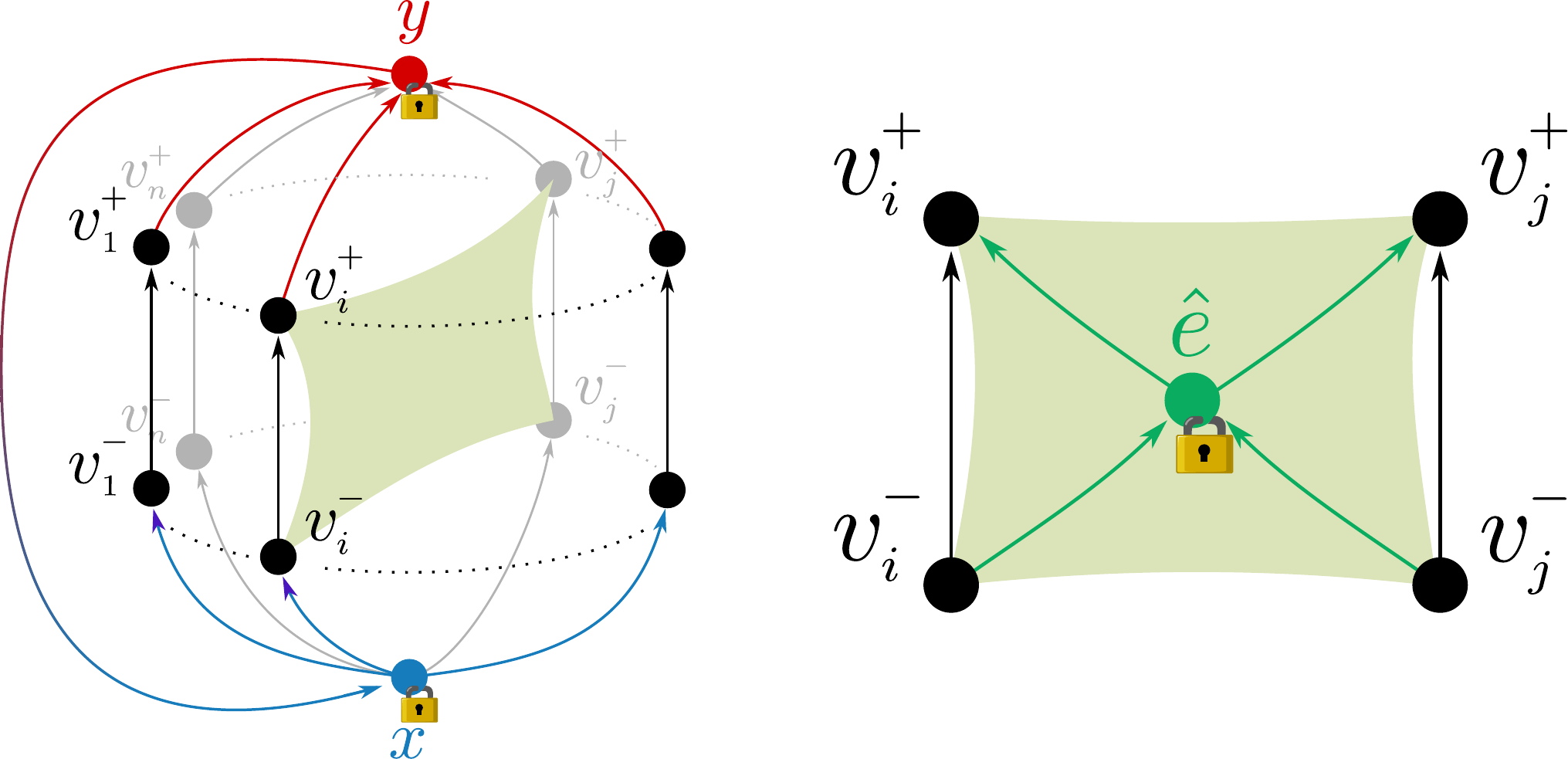}
  \end{center}

  \caption{An illustration of the arcs in the reduced instance of {\strongcontract}. The second figure only contains the arcs of the edge-selection gadget corresponding to the edge $e=(v_i,v_j)\in E(G)$. Vertices with a padlock have additional $k+1$ pendant
  vertices with arcs in both directions.}
  \label{fig:path_contraction_reduction}
\end{figure}

\noindent	
We now consider the converse direction. Suppose that $(D,k)$ is a yes-instance of {\strongcontract} and let $S=\{a_1,\dots, a_k\}$ be a solution for this instance. We require the following claim.	
	
	\begin{claim} For every edge $e=(u,v)\in E(G)$, $|S\cap \{(u^-,u^+),(v^-,v^+)\}|\leq 1$.  Furthermore, $S\subseteq \{(v^-,v^+) \mid v\in V(G)\}$.
	\end{claim}
	
	\begin{proof}
	For the first statement, suppose to the contrary that $S$ contains both the arcs $(u^-,u^+)$ and $(v^-,v^+)$ for some $e=(u,v)\in E(G)$. Then, observe that in the graph $D \contract S$, the arcs in the set $\cB_e$  
	are absent. Since $\cB_e$ contains all arcs incident to $\hat e$ except the ones incident to $\hat e_i$ for $i\in [k+1]$, this disconnects the undirected graph underlying $D \contract S$, implying that $S$ is not a solution, a contradiction.

	For the second statement, we argue that no arc incident to $x,y$ or $\{\hat e \mid e\in E(G)\}$ can be in $S$. Suppose to the contrary that for some $i\in [k+1]$, the arc $(x,x^i)\in S$. Then, the arcs from $x$ to $\{v^- \mid v\in V(G)\}$ are all absent from $D'$, implying that $D'$ is not strongly connected, a contradiction. On the other hand, if for some $v\in V(G)$, we path-contract the arc $(x,v^-)$, the arc from $x$ to $x^i$ is absent in $D \contract S$ for every $i\in [k+1]$. Since $|S|\leq k$, there is at least one $i\in [k+1]$ such that the arc $(x,x^i)$ is not in $S$. Since the arc $(x,x^i)$ is absent from $D \contract S$, it follows that it is not strongly-connected, a contradiction. Finally, if $S$ contains the arc $(y,x)$, the arc $(x^i,x)$ is not in $D \contract S$ for any $i\in [k+1]$, implying that it is not strongly-connected for the same reason as that in the previous case. Hence, we conclude that no arc incident on $x$ is in $S$. The argument for $y$ is analogous and hence we do not address it explicitly.

	Suppose that for some $e=(u,v)\in E(G)$ and $i\in [k+1]$, there is an arc in $\{(\hat e,\hat e_i), (\hat e_i, \hat e)\}$ which is in $S$. Observe that in the former case, the arcs $(\hat e,u^+)$ and $(\hat e,v^+)$ are absent in $D \contract S$, implying that the new vertex is a sink, a contradiction. In the latter case, the new vertex is a source, a contradiction. Now, suppose that $S$ contains an arc in $\cB_e$. Then, for some $i\in [k+1]$, the vertex $\hat e_i$ is left as a source or sink  in $D \contract S$, a contradiction. This completes the proof of the claim.
	\end{proof}

	\noindent
	The claim above implies that if $X$ is a solution for the reduced instance of {\strongcontract}, then the set $S$ of arcs corresponds  independent set in $G$. In other words, $(G,k)$ is a yes-instance of {\sc Independent Set}.
	This proves the correctness of the reduction and  completes the proof of the theorem. 
\end{proof}

\noindent
We can prove a similar result for the {\strongnodedel} problem. The problem is formally defined as follows. 

\begin{problem}[$k$]{\strongnodedel}
  \Input & A strongly connected digraph $D$ and an integer $k$. \\
  \Prob  & Is there a vertex set $S$ of size (exactly) $k$ 
             such that the graph  $D-S$  is  strongly connected? 
\end{problem}

\noindent
We have to require ``exactly $k$'' rather than ``at least $k$'' since
otherwise we could delete all but one vertices of $D$ and get a trivially 
strongly connected digraph.

\thmnodedelmain*

\begin{proof}
We will again use a reduction from  {\sc Independent Set}. Let $G$ be a graph, an input of {\sc Independent Set} with parameter $k$. 
We first reduce {\sc Independent Set} to {\sc Vertex-deletion Preserving Connectivity with Undeletable Vertices}: Given a connected graph $H$ with some vertices marked and parameter $k$, is there $k$ unmarked vertices in $H$ whose deletion keeps $H$ connected? To construct $H$, start from $G$ with all vertices unmarked. Subdivide every edge of $G$ with a marked vertex. Add another marked vertex $x$ with edges to all unmarked vertices. 
It is easy to see that the reduction is correct since deleting two unmarked vertices in $H$ which are adjacent in $G$ leaves the corresponding subdivision vertex isolated. 

Now we reduce {\sc Vertex-deletion Preserving Connectivity with Undeletable Vertices} to \strongnodedel{}. Replace every edge $uv$ of $H$ by arcs $uv$ and $vu$,  unmark every marked vertex $w$ of $H$ and replace it by a directed cycle of length $k+2$ containing $w$ (all other vertices of the cycle are new).  Denote the resulting digraph by $D$; note that it is strongly connected. 
To see the correctness, it suffices to observe that we cannot delete less than $k+1$ vertices of any directed cycle of length $k+2$ and keep $D$ strongly connected. This completes the proof of the theorem. 
\end{proof}

\section{Edge deletion to biconnected graphs}\label{sec:biconnectivity}

In this section, we present our {\sf FPT} algorithm for the {\twoconndel} problem
on undirected graphs. Recall that the problem is defined as follows:

\begin{problem}[$k$]{\twoconndel}
  \Input & A biconnected graph $G$,  $k\in \naturals$, 
           $w^*\in \reals_{\geq 0}$ and a function $w:E(G)\to \reals_{\geq 0}$. \\
  \Prob  & Is there a set $S \subseteq E(G)$ of size at most $k$ such that 
           $G - S$ is biconnected and $w(S) \geq w^*$? 
\end{problem}

\noindent
We refer to a set $S\subseteq E(G)$ such that $G-S$ is biconnected as a {\em
biconnectivity deletion set} of $G$. For an instance $(G,k,w^*,w)$ of
{\twoconndel} and a biconnectivity deletion set $S$ of $G$, we say that $S$ is
a {\em solution} if $|S|\leq k$ and $w(S)\geq w^*$. The main result of this
section is the following.

\thmbiconnected*

\noindent
We will first make a short digression in order to define the notion of
critical edges and list certain structural properties that will be required in
this and the following section.

\subsection{Properties of critical edges}

\begin{definition}\label{def:critical}
	We denote by $\kappa(G)$ the vertex-connectivity of a graph $G$. Let $G$ be a $\rho$-vertex connected   graph. 
	An edge $e\in E(G)$ is called $\rho$-{\em critical} if $\kappa(G-e)<\rho$.	We denote by ${\sf Critical^\rho_G(e)}$ the subset of $E(G)$ comprising edges which are $\rho$-critical in $G-e$ but not in $G$. 
	We denote by $\critical^\rho_G(\emptyset)$ the set of edges which are already $\rho$-critical in $G$. 
	In all notations, we ignore the explicit reference to $G$ and $\rho$ when these are clear from the context.
	We say that $e$ is $\rho$-{\em  critical for} a pair of vertices $u,v$ in $G$ if $u$ and $v$ are non-adjacent and $e$ participates in every  $u$-$v$ flow of value $\rho$ in $G$. 
\end{definition}

\noindent
The following lemma gives a useful structural characterization of edges which
become $\rho$-critical upon the deletion of a particular edge of the graph.

\begin{figure}[t]
  \begin{center}\includegraphics[scale=0.4]{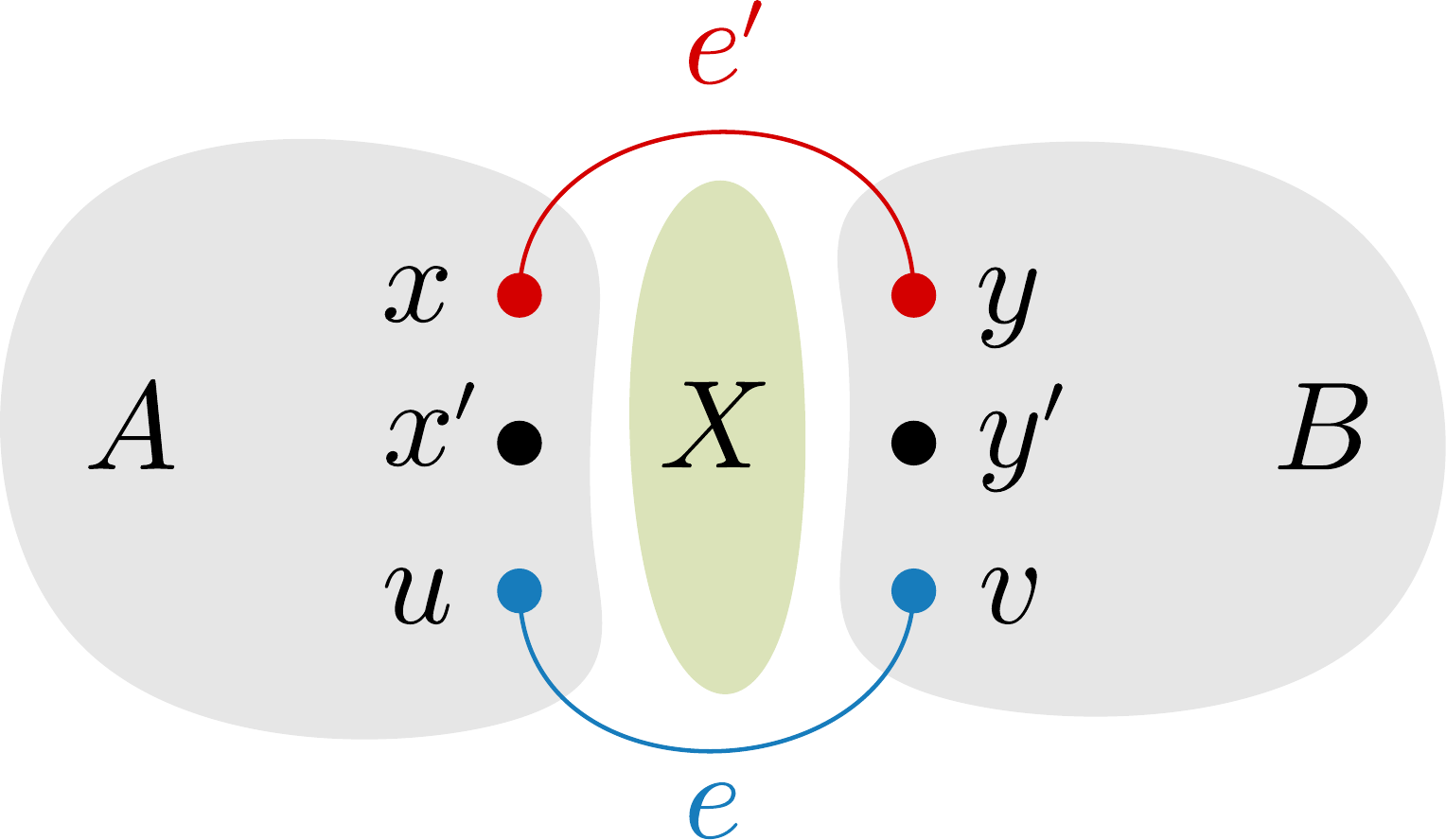}
  	
  \end{center}

  \caption{An illustration of the edges $e,e'$ and the sets $A$ and $B$ in the proof of Lemma~\ref{lem:full_existence}.  Observe that there are no more edges with one endpoint each in $A$ and $B$. }
  \label{fig:first_separation}
\end{figure}

\begin{lemma}\label{lem:full_existence}
	Let $G$ be a $\rho$-vertex connected graph. Let $e=(x,y)$ and $e'=(u,v)$ be 
  distinct non-critical edges, \ie they are not in $\critical^\rho_G(\emptyset)$. Then the following are equivalent:
	
	\begin{enumerate}\item $e'\in \critical^\rho_G(e)$. 
	\item There is a pair of vertices $x',y'\in V(G)$ and a mixed $x'$-$y'$ cut $X$ of size $\rho$ in $G-e$, where $e'\in X$.
\item $e'$ participates in  every $x$-$y$ flow of value $\rho$ in $G-e$.
\end{enumerate}
\end{lemma}

\begin{proof} We prove that $(1) \Rightarrow (2)$, $(2)\Rightarrow (3)$ and $(3)\Rightarrow (1)$.
 Consider the first implication. Since $e'\in \critical^\rho_{G}(e)$, it must be the case that there are vertices $x',y'$ such that $e'$ is $\rho$-critical for the pair $x',y'$. That is,  $x'$ is non-adjacent to $y'$ and $e'$ is $\rho$-critical for the pair $x',y'$ in $G'=G-e$. That is, $e'$ participates in every $x'$-$y'$ flow of value $\rho$ in $G-e$. As a result, the maximum value of any $x'$-$y'$ flow in $G-e-e'$ is precisely $\rho-1$. By Menger's theorem, this implies the presence of a set $S\subseteq V(G)$ of size $\rho-1$ which intersects all $x'$-$y'$ paths in $G-e-e'$. Setting  $X=S\cup \{e'\}$  completes the argument for the first implication.

Consider the second implication. Let $A=R_{G-e}(x',X)$ and let $B=V(G)\setminus (A\cup S)$ (see Figure~\ref{fig:first_separation}), where $S=X\setminus \{e'\}$. Observe that if $u,v\in A$ or $u,v\in B$, then $X\setminus \{e'\}$ would be an $x'$-$y'$ separator of size $\rho-1$ in $G-e$, a contradiction to $G-e$ being $\rho$-vertex connected. Hence, it must be the case that 
either $u\in A$ and $v\in B$ or vice-versa. We assume without loss of generality that $u\in A$ and $v\in B$.

 By a similar argument,  since $e'$ is not $\rho$-critical in $G$ but is $\rho$-critical in $G-e$, it must be the case that the edge $e$ also has one endpoint in $A$ and one endpoint in $B$.  Again, we assume without loss of generality that  $x\in A$ and $y\in B$.	
	But now, observe that in the graph $G-e-S$, every $x$-$y$ path contains the edge $e'$.   Since $|S|=\rho-1$ and $G-e$ is $\rho$-connected,  we conclude that $e'$ participates in every $x$-$y$ flow of value $\rho$ in $G-e$. This completes the argument for the second implication.
	
	For the final implication, observe that while $G-e$ is $\rho$-connected, the fact that $e'$ participates in every $x$-$y$ flow of value $\rho$ in $G-e$ implies that $e'$ is $\rho$-critical in $G-e$. Since $e'$ is by definition, not $\rho$-critical in $G$, we conclude that $e'\in \critical^\rho_G(e)$. This completes the proof of the lemma.	
\end{proof}

\subsection{The {\sf FPT} algorithm for {\twoconndel}}

In this section, we will prove Theorem~\ref{thm:biconnected} by giving an algorithm for a more general version of the {\twoconndel} problem where the input also includes a set $E^\infty\subseteq E(G)$ and the objective is to decide whether there is a solution disjoint from this set. Henceforth, instances of {\twoconndel} will be of the form $(G,k,w^*,w,E^\infty)$ and any solution $S$ is required to be disjoint from $E^\infty$. We will refer to edges of $E(G) \setminus E^\infty$ as \emph{potential solution edges}. We say that a potential solution edge is \emph{irrelevant} if 
 either the instance has no solution, or has a solution that does not contain $e$. For an instance $I=(G,k,w^*,w,E^\infty)$ and $r\in \naturals$, we denote by $\heavy_I(r)$ the \emph{heaviest} $r$ potential solution edges of $G$ with respect to the function $w$. While this set is not necessarily unique (if multiple edges have the same weight, i.e., the same image under the function $w$), we will define $\heavy_I(r)$ as the first $r$ edges of \emph{a fixed  arbitrarily chosen ordering} of the edges of $G$ in non-increasing order of their weights. 
  If $I$ is clear from the context, we simply write $\heavy(r)$ when referring to $\heavy_I(r)$.

Since we will only be dealing with biconnected graphs in this section, we will also  drop the explicit reference to $\rho$ in the notations from Definition~\ref{def:critical}. For instance, when we say that an edge is critical (non-critical), we imply that it is 2-critical (not 2-critical respectively). 
%
%
%
Observe that no edge from the set $\critical_G(\emptyset)$ can be part of a
solution. As a result, we assume without loss of generality that for any
instance $(G,k,w^*,w,E^\infty)$, the set $\critical_G(\emptyset)$ is contained
in $E^\infty$.  Furthermore, since the edges in $E^\infty$ can never be part
of a solution, we assume without loss of generality that for every edge $e\in
E^\infty$, $w(e)=0$. The proof of Theorem~\ref{thm:biconnected} is based on
the following lemma which states that either a) the number of potential
solution edges in the instance is already bounded polynomially in $k$, or b) a
`small' set of the heaviest edges in the instance must intersect a solution, or
c) there is  an irrelevant edge which can be found in polynomial time. For 
ease of presentation, let use define the polynomial $\mu(x) := 20x^3+46x^2+x$
for the rest of this section.

\begin{lemma}\label{lem:biconnected_bound_on_non_critical_edges} Let $I=(G,k,w^*,w,E^\infty)$ be an instance of {\twoconndel}. If  
$|E(G)\setminus E^\infty|> \mu(k)$, then the set $\heavy(\mu(k))$ contains either a solution edge or an irrelevant edge which can be computed in polynomial time. 
\end{lemma}

\noindent
Given Lemma~\ref{lem:biconnected_bound_on_non_critical_edges},
Theorem~\ref{thm:biconnected} is proved as follows. Let
$I=(G,k,w,w^*,E^\infty)$ be an instance of {\twoconndel}. If the  number of
potential solution edges in this instance is already bounded by $\mu(k)$, then
we simply enumerate all $k$-sized subsets of this set (there are
$2^{O(k\log k)}$ choices) and check in polynomial time whether one of these subsets is a
solution. Otherwise, we invoke
Lemma~\ref{lem:biconnected_bound_on_non_critical_edges} and either correctly
conclude that the set $\heavy(\mu(k))$ contains a solution edge,  or we
compute an irrelevant edge $e$ in polynomial time. In the first case we branch
on the set $\heavy(\mu(k))$, reduce the budget $k$ by 1 and the target weight
$w^*$ accordingly and recursively solve the resulting instance. In the second
case, we add the edge $e$ to the set $E^\infty$ (thus decreasing the set of
potential solution edges) and repeat.

\begin{remark}\label{rem:diff_strategy}
 There is also an alternative strategy to the above, as follows. Let $S$ be the set of all edges of weight at least $w^*/k$. 
  Clearly $S$ must be non-empty and any solution must intersect $S$. If $|S| \leq \mu(k)$, then we branch on $S$ as above. 
  Otherwise, we will be able to either find a biconnectivity deletion set $S' \subseteq S$ with $|S'|=k$ or an irrelevant edge in $S$ as in Lemma \ref{lem:biconnected_bound_on_non_critical_edges}. 
  In the former case, $S'$ is already a solution; in the latter case, we proceed according to 
  the strategy above. Thus, this alternative strategy yields a slightly simpler proof, contains one less branching step and will be used in the kernelization algorithm in  Subsection~\ref{subsubsec:kernel}. 
  On the other hand, the strategy above does not explicitly depend on $w^*$, and therefore always gives a maximum-weight solution.
  In either case, the main technical challenges in the {\sf FPT} algorithm are exactly the same.
\end{remark}

\noindent
The rest of this section is devoted to proving Lemma~\ref{lem:biconnected_bound_on_non_critical_edges}. In order to do so, we will present a greedy algorithm that runs in polynomial time and, assuming $|E(G) \setminus E^\infty|>\mu(k)$, will either produce a biconnectivity deletion set of size $k$ contained strictly within $\heavy(\mu(k))$, or it will identify an irrelevant edge. In the former case, we will argue that this implies that there is always a solution intersecting $\heavy(\mu(k))$. 
More precisely, the algorithm will  delete one potential solution edge from $\heavy(\mu(k))$ at a time (while preserving biconnectivity), and will trace in each step the number of edges of $\heavy(\mu(k))$ that become critical due to the removal of such an edge $e$, i.e., the size of the set $\critical_{G'}(e)\cap \heavy(\mu(k))$ where $G'$ is the subgraph of $G$ remaining after deleting the edges before $e$.  We will then show that if $|\critical_{G'}(e)\cap \heavy(\mu(k))|\geq \frac{\mu(k)-k}{k}$, then $G$ contains a special configuration from which we can either recover the required biconnectivity deletion set or identify an irrelevant edge.

\subsubsection{Preliminary results}\label{subsubsec:greedy}

From now on, we assume that the given instance has more than $\mu(k)$ potential solution edges and begin by proving the following lemma which shows that if we find \emph{some} biconnectivity deletion set of size $k$ within $\heavy(\mu(k))$, then there is a solution intersecting $\heavy(\mu(k))$.

\begin{lemma}\label{lem:greedy_yes_global_yes}
Let $I=(G,k,w^*,w,E^\infty)$ be an instance of {\twoconndel} and let $S\subseteq \heavy(\mu(k))$ be a biconnectivity deletion set of size $k$. If $I$ is a yes-instance, then there is a solution for $I$ intersecting the set $\heavy(\mu(k))$.
\end{lemma}

\begin{proof}
	Suppose that this is not the case and let $S'$ be a biconnectivity deletion set of size at most $k$ such that  $w(S')\geq w^*$. Note that $S'$ is disjoint from $\heavy(\mu(k))$ and $S$ is contained in $\heavy(\mu(k))$. Since $|S'|\leq |S|$, we infer that $w(S)\geq w(S')$, a contradiction to our assumption that there is no solution intersecting $\heavy(\mu(k))$. 
	This completes the proof of the lemma.
\end{proof}

\noindent
Let $\hat S=\{f_1,\dots, f_r\}\subseteq \heavy(\mu(k))$ be a set greedily constructed as follows. The edge $f_1$ is the heaviest potential solution edge. That is, $w(f_1)\geq w(e)$ for every $e\in E(G)\setminus E^\infty$. For each $2\leq i\leq r$, $f_i$ is the heaviest edge of $\heavy(\mu(k))$ which is \emph{not} critical in $G-\{f_1,\dots, f_{i-1}\}$. We terminate this procedure after $k$ steps if we manage to find edges $\{f_1,\dots, f_k\}$ or earlier if for some $r<k$, every edge of $\heavy(\mu(k))$  is critical in $G-\{f_1,\dots, f_r\}$.

Observe that by definition, $\hat S$ is a biconnectivity deletion set. Therefore, if $r=k$, then Lemma \ref{lem:greedy_yes_global_yes} implies that if there is a solution for the given instance, then there is one intersecting $\heavy(\mu(k))$ (as required in Lemma~\ref{lem:biconnected_bound_on_non_critical_edges}).  On the other hand, suppose that $r<k$. For each $i\in [r]$, we denote by $\hat S_{i}$, the set $\{f_1,\dots, f_i\}$ and by $\hat S_{0}$, the empty set. Recall that we have already assumed that the number of potential solution edges is greater than $\mu(k)=20k^3+46k^2+k$. As a result, we have the following observation. 

\begin{observation}\label{obs:exists_good_index}
	There is an $i\in [r]$ such that $G-\hat S_{i}$ is biconnected and   $|\critical_{G-\hat S_{i-1}}(f_i)\cap \heavy(\mu(k))|\geq \frac{\mu(k)-k}{k} = 20k^2+46k$.
\end{observation}

\noindent
Let $i\in [r]$ be the index referred to in this observation. In the rest of
the section, we let $\hat F=\hat S_{i-1}$, $e= f_i=(x,y)$ and $G'=G-\hat F$.
The following observation is a straightforward consequence of
Lemma~\ref{lem:full_existence} in our setting.

\begin{observation}
\label{obs:full_existence_3}
  Let $G'$ and $e$ be as above. Then, 
  $\kappa_{G'-e}(x,y)=2$ and for any $x$-$y$ flow $\cH=\{H_1,H_2\}$ of value 2 in $G'-e$, 
  the following holds.  
  \begin{enumerate}
  \item Every edge of $\critical_{G'}(e)$ is critical for the pair $x, y$ in $G'-e$ and hence lies on $H_1$ or $H_2$.
  \item For every edge $e_1 \in \critical_{G'}(e)$, say on $H_1$, there
    is at least one vertex $v$ on $H_2$ such that $\{e_1,v\}$ is a mixed $x$-$y$ cut in $G'-e$.
  \end{enumerate}
\end{observation} 
  
\noindent
We refer to the vertex $v$ above as a \emph{partner vertex} of $e_1$, and
refer to the set of all partner vertices of $e_1$ as the \emph{partner set} of
$e_1$ and denote this set by $\partner^e_{G'-e}(e_1)$. We do not explicitly
refer to $H_1$ or $H_2$ in this notation  because these will always be clear
from the context. We will also drop the explicit reference to $G'$ and $e$
when these are clear from the context.

From Observation~\ref{obs:exists_good_index} and
Obsevation~\ref{obs:full_existence_3}, we now conclude the following:

\begin{observation}\label{obs:rich_flow_exists}
	Let $G'$ and $e$ be as above. There is an $x$-$y$ flow $\cP=\{P_1,P_2\}$ of value 2 in $G'-e$ such that $|E(P_1)\cap \critical_{G'}(e)\cap \heavy(\mu(k))|\geq 10k^2+23k$.
	\end{observation}
	
\noindent
Henceforth, we work with this fixed $x$-$y$ flow $\cP=\{P_1,P_2\}$ in $G'-e$.

\begin{figure}[t]
  \begin{center}\includegraphics[scale=.35]{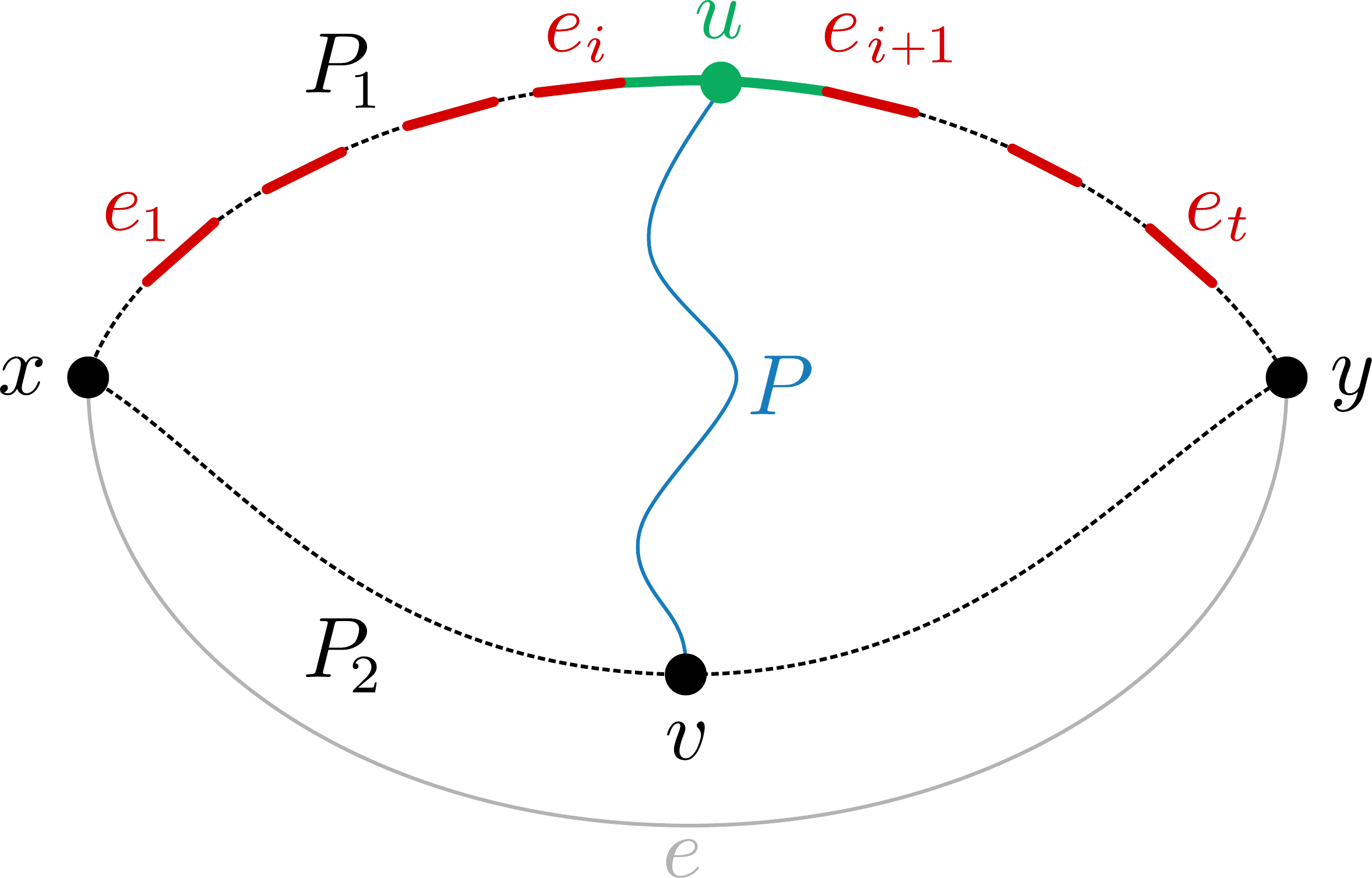}
  	
  \end{center}

  \caption{An illustration of $\segment[i,i+1]$ (the subpath of $P_1$ between $e_i$ and $e_{i+1}$) and $P$ which is a nice path for this segment.}
  \label{fig:segment}
\end{figure}

\begin{definition}Let $G'$,$e=(x,y),P_1$ and $P_2$ be as above. 
	Let $e_1, \ldots, e_t$ be some subset of $10k^2+23k$ edges in $\critical_{G'}(e) \cap \heavy(\mu(k))\cap E(P_1)$
  in the order in which they appear when traversing $P_1$ from $x$ to $y$ (see Figure~\ref{fig:segment}), where $e_i=(u_i,v_i)$ and we may have $v_i=u_{i+1}$. \begin{itemize}\item For $i \in [t-1]$, we  refer to the subpath of $P_1$
  from $v_i$ to $u_{i+1}$ inclusively as $\segment[i, i+1]$ of $P_1$ with the explicit reference to $P_1$ dropped when clear from the context. \item A path $P$   with endpoints $u, v \in V(P_1) \cup V(P_2)$
  but internally vertex-disjoint from $P_1$ and $P_2$ is said to be a {\em nice path} for $\segment[i,i+1]$ if $u$ is contained in $\segment[i,i+1]$ and $v\in V_{\sf int}(P_2)$. 
  \end{itemize}
\end{definition}

\noindent
When $e$, $P_1$, $P_2$ are as in the definition above, we write $w \leq w'$
for vertices $w, w' \in V(P_2)$ if either $w=w'$ or $w$ is encountered
before $w'$ when traversing $P_2$ from $x$ to $y$, and similarly for
vertices on $P_1$. Furthermore, for a set $Q\subseteq V(P_2)$ and vertex
$w\in V(P_2)$, we say that $w< Q$ ($w>Q$) if $w<q$ ($w>q$ respectively) for
every $q\in Q$. We need the following crucial structural lemma regarding the
structure of any path with endpoints in $P_1$ and $P_2$ but which is
otherwise disjoint from these two paths.

\begin{lemma} \label{lem:segmentconnections}
  Let $G',e=(x,y),P_1, P_2$, $e_1,\dots, e_t$ be as above.
   Let $P$ be a path in $G'$ with endpoints $u, v \in V(P_1) \cup V(P_2)$
  but internally vertex-disjoint from $P_1$ and $P_2$.  Then the following statements hold.
  \begin{enumerate}
  \item If $u, v \in V(P_1)$, then either $u, v \leq u_1$, or $u, v \geq v_t$, or
    there is a $j\in [t-1]$ such that $u,v$ lie on $\segment[j,j+1]$.
  \item If $u, v \in V(P_2)$, then the subpath of $P_2$ from $u$ to $v$ 
      is  internally vertex-disjoint from the set $\bigcup^t_{i=1} \partner^e_{G'}(e_i)$.
  \item If $u \in V(P_1)$ and $v \in V(P_2)$, where $u$ lies in $\segment[i,i+1]$,
    then for every $j \leq i$ and every vertex $w \in \partner(e_j)$ we have $w \leq v$,
    and for every $j \geq i+1$ and every vertex $w \in \partner(e_j)$ we have $w \geq v$. \label{item:nicep2order}
    \item  For every $i\in [t-1]$, there is at least one nice path for $\segment[i,i+1]$, and for any $i\neq j\in [t-1]$, and paths $P$ and $P'$ which are nice paths for $\segment[i,i+1]$ and $\segment[j,j+1]$ respectively, 
    $P$ and $P'$ are internally vertex-disjoint. 
    \label{item:pathsp1p2}
  \end{enumerate}
\end{lemma}
\begin{proof}
For the first statement, suppose that there is an index $i \in [t]$ such that $u \leq u_i$ but $v \geq v_i$. Since $e_i$ is critical in $G'-e$, 
Observation~\ref{obs:full_existence_3} implies that there is a mixed $x$-$y$ cut $X=\{e_i,q\}$ where the vertex $q$ lies on $P_2$. However, the graph induced on $V(P_1)\cup V(P)$ contains an $x$-$y$ path disjoint from $X$, a contradiction. This completes the argument for the first statement. 

The argument for the second statement is similar. Suppose to the contrary that there is an $i\in [t]$ and a vertex $w\in \partner(e_i)$ such that the subpath of $P_2$ from $u$ to $v$ contains the vertex $w$. Recall that due to Observation~\ref{obs:full_existence_3}, the set $X=\{e_i,w\}$ is a mixed $x$-$y$ cut in $G'-e$. But then $u \in R_{G'-e}(x,X)$, $v \in R_{G'-e}(y,X)$, and the path $P$ is disjoint from $X$, which contradicts $X$ being an $x$-$y$ cut. 

For the third statement, suppose that there is a $j\leq i$ and $w\in \partner(e_j)$ such that $w>v$. Let $X=\{e_j,w\}$. Due to Observation~\ref{obs:full_existence_3}, we know that $X$ is a mixed $x$-$y$ cut in $G'-e$.  Let $A=R_{G'-e}(x,X)$ and $B=R_{G'-e}(y,X)$. Since $u\in \segment[i,i+1]$ and $j\leq i$, it follows that $u\in B$. Similarly, since $w>v$, it must be the case that $v\in A$.  
As above, we find that $P$ is a path disjoint from $X$ connecting $A$ and $B$, which contradicts that $X$ is an $x$-$y$ cut.
The argument for the case when there is a $j\geq i$ and $w\in \partner(e_j)$ such that $w<v$ is analogous. 
%
%

  For the first part of the final statement, assume for a contradiction that for some $i\in [t-1]$, the path $\segment[i,i+1]$ does not have a nice path. Recall that
   $e_i$ is not critical in $G'$ but is critical in $G'-e$. 
  Therefore, there is a $u_i$-$v_i$ flow of value 2 in the graph $G'-e_i$; let $\cH=\{H_1,H_2\}$ be such a flow. If $H_1$ or $H_2$ intersects the internal vertices of $P_2$, then 
  this implies the presence of a nice path for $\segment[i, i+1]$. 
  Hence, we assume that this does not happen. We also conclude that $e$ must occur in $H_1$ or $H_2$. 
  Indeed, observe that $e$ is critical in $G'-e_i$, since $G'-e_i$ is biconnected but $G'-e-e_i$ is not. Hence $e \in \critical_{G'}(e_i)$, and by Lemma~\ref{lem:full_existence}, $e$ must participate in $\cH$. 
  We may assume without loss of generality that $H_1$ contains the edge $e$.
  But now $H_2$ is a path from $u_i$ to $v_i$ in $G'$, disjoint from both $e$, $e'$, and $V_{\sf{int}}(P_2)$. 
  Clearly, $H_2$ contains a subpath $P$ in contradiction to the first statement of the present lemma.
  We conclude that for every $i\in [t-1]$, $\segment[i,i+1]$  has a nice path. 

  For the second part of the last statement, let $P$ and $P'$ be paths as described which are not internally vertex-disjoint. Then $P \cup P'$ contains a walk, and therefore also a path, with endpoints in $V(P_1)$, internally vertex-disjoint from $V(P_1) \cup V(P_2)$, and with endpoints in distinct segments on $P_1$, in contradiction with the first statement of this lemma.
  This completes the argument for the last statement and hence the proof of the lemma.    
%
%
%
%
%
%
%
%
\end{proof}


\noindent

Let us now consider how partner sets can intersect.

\begin{observation} \label{obs:partner_switching} Let $G',e=(x,y),P_1, P_2$, $e_1,\dots, e_t$ be as above. 
  Let $e_i, e_j$ be a pair of edges, $1 \leq i < j \leq t$, 
  let $w_1, \ldots, w_r$ be the partner vertices of $e_i$ in the order they appear on $P_2$,
  and let $w_1', \ldots, w_s'$ be the partner vertices of $e_j$ in the order they appear on $P_2$.
  Then $w_i \leq w_j'$ for every $i \in [r]$, $j \in [s]$. In particular, the set
  $\partner(e_i) \cap \partner(e_j)$ can consist of at most one vertex $w$, which {\em must} then 
be the last vertex of $\partner(e_i)$ and the first vertex of $partner(e_j)$ which is encountered when traversing $P_2$ from $x$ to $y$. 
\end{observation}

\begin{proof} 
  Let $w_i\in \partner(e_i)$ and $w_j'\in \partner(e_j).$ By Lemma~\ref{lem:segmentconnections} (4), $\segment[i,i+1]$ has a nice path $P$ with endpoints $u\in V(P_1)$ and $v\in V(P_2)$. By Lemma \ref{lem:segmentconnections} (3), we have $w_i\le v\le w_j'$. Thus, $w_i\le w_j'$ and the claim follows.
\end{proof}

\noindent
Thus, there is a well-defined first and last element for each partner set and these two elements (they may coincide) define a subpath of $P_2.$ Furthermore, the two subpaths corresponding to the partner sets of any two  critical edges on $P_1$ do not have a `strict' overlap and can only intersect in one vertex -- their respective endpoints.


Having identified some of the structure in the graph, we now proceed to examine two cases.
Recall that by Observation \ref{obs:rich_flow_exists},  the path
$P_1$ contains at least $10k^2+23k$ edges of $\critical_{G'}(e)\cap \heavy(\mu(k))$. 
We will consider one of two cases: either there is a sufficiently large number of distinct partner sets, or there is a sufficiently large number of critical edges with identical partner sets. 
We show how to handle each case in turn.

\subsubsection{Many distinct partner sets}\label{subsubsec:case_analysis}

We first handle the first case, by formally arguing that if there are
sufficiently many distinct partner sets, then $\heavy(\mu(k))$ contains a
solution edge. We begin with an observation about connectivity.

\begin{lemma} \label{lem:segmentdetour}
  Let $G',e=(x,y),P_1, P_2$, $e_1,\dots, e_t$ be as above.
  For each $i\in [t]$, there is a pair of internally vertex-disjoint
  $u_i$-$v_i$ paths $P_a$, $P_b$ in $G'-e_i$ as follows.
 \begin{enumerate}
 \item $P_a$ contains the edge $e$. Additionally,
   if $i>1$, then $P_a$  either contains $e_{i-1}$ or intersects $\partner(e_{i-1})$, 
   and if $i<t$, then $P_a$  either contains $e_{i+1}$ or intersects $\partner(e_{i+1})$.
  \item $P_b$ has an endpoint each in $\segment[i-1,i]$ and $\segment[i,i+1]$, but does not intersect $P_1$
    anywhere else except in these segments, \label{item:noothersegments}
  \item $P_b$ contains the set $\partner(e_i)$ and is disjoint from the set $\partner(e_j) \setminus \partner(e_i)$
    for any $j \in [t]$, $j \neq i$, except possibly the vertices of $\bigcup_j \partner(e_j)$
    immediately preceding and succeeding $\partner(e_i)$ on $P_2$.
  \end{enumerate} 
\end{lemma}

\begin{proof}
  Let $(P_a, P_b)$ be a pair of internally vertex-disjoint $u_i$-$v_i$ paths in $G'-e_i$. This exists since $e_i$ is not critical in $G'$.
  Let $w \in \partner(e_i)$. Since $\{w,e_i\}$ is a mixed $x$-$y$ cut in $G'-e$, it follows that $\{w,e\}$ is a mixed $u_i$-$v_i$ cut in $G'-e_i$. 
  Hence one path, say $P_a$, must pass through $e$, and the other must pass through $w$. Since $w$ was arbitrarily chosen,
  we find that $P_b$ contains every vertex of $\partner(e_i)$. Next, assume $i>1$ and let $w'$ be the largest vertex in $\partner(e_{i-1})$ in the order $<$.
  Then $\{w', e_{i-1}\}$ is a mixed $x$-$y$ cut in $G'-e$, thus $P_a$ must pass through either $e_{i-1}$ or $w'$ on the way from $u_i$ to $x$. 
  The dual argument holds for $e_{i+1}$ if $i<t$. This covers the first property. 
  For the second and third properties, consider again the mixed cut $\{e_{i-1}, w'\}$. Since $P_b$ contains $u_i$ and $v_i$, 
  both of which are on the same side in the above cut, $P_b$ passes through the cut an even number of times; since $P_a$ intersects
  the cut, $P_b$ cannot pass through the cut and so cannot intersect $P_1$ in any segment before $\segment[i-1,i]$,
  nor $P_2$ in any vertex before $w'$. (Note that $P_b$ may intersect $w'$, but it cannot intersect any vertex
  on the other side of the cut.) An analogous argument holds for $e_{i+1}$ if $i<t$. 
\end{proof}

\noindent
We now state and prove the lemma which handles the first case, \ie there are a
sufficiently large number of distinct partner sets.

\begin{lemma}
  \label{lem:manypartnersets} 
  
  Let $G',e=(x,y),P_1, P_2$, $e_1,\dots, e_t$ be as above.
    Assume that there are more than $3k$ distinct partner sets for the edges $e_1,\dots, e_t$. 
  Then the instance $(G,k,w^*,w,E^\infty)$ has a solution intersecting $\heavy(\mu(k))$.
\end{lemma}
\begin{proof}
  Let $Z=\{e_1,\dots, e_t\}$ and let $e_{i_1}$, \ldots, $e_{i_{3k+1}}$, be a subset of $Z$ such that 
  for every $1\leq p<q\leq 3k+1$, (a)  $i_{q}>i_p$  
  and (b) $\partner(e_{i_p})\neq \partner(e_{i_{q}})$.
  Let $S=\{e_{i_1}, e_{i_4}, \ldots, e_{i_{3k-2}}\}$. Clearly $|S|=k$;
  we claim that $S$ is a biconnectivity deletion set for $G'$. 
  
  To see this, let $e_{i_j}=(u_{i_j},v_{i_j})$ be an arbitrary edge of $S$,
  and let $P_a$, $P_b$ be $(u_{i_j},v_{i_j})$-paths given by  Lemma~\ref{lem:segmentdetour}.
  Then the path $P_b$ remains in $G'-S$; we will reconfigure $P_a$ to be disjoint from $S$.
  We will create a path $P=P_x + (x,y) + P_y$, by separately providing a path $P_x$ from $u_{i_j}$ to $x$
  and a path $P_y$ from $v_{i_j}$ to $y$ which are disjoint from $P_b$ and
  neither of which contains the edge $e=(x,y)$. If $j=1$, then $e_{i_j}$ is the first edge of $S$ along $P_1$
  and we may simply use $P_a$ from $u_{i_j}$ to $x$ as $P_x$, so assume $j>1$.
  If $P_a$ intersects $\partner(e_{i_j-1})$, then we may produce $P_x$ by continuing along $P_2$ to $x$.
  Otherwise $P_a$ uses the edge $e_{i_j-1}$. In this case, produce $P_x$ by continuing along $P_1$
  to $\segment[i_{j-3}, i_{j-3}+1]$, follow a nice path from this segment to $P_2$, and continue along $P_2$ to $x$.
  
  We argue that the resulting path $P_x$ is disjoint from $P_b$. If $j=1$, then the claim is trivial.
  If $P_a$ intersects $\partner(e_{i_j-1})$, then recall that $P_a$ and $P_b$ are internally disjoint,
  $P_b$ intersects $\partner(e_{i_j})$, and $\partner(e_{i_j-1}) \leq \partner(e_{i_j})$. 
  Thus $P_x$ lies entirely before $P_b$ on $P_2$. Otherwise, $P_x$ uses a nice path $P$ from $\segment[i_{j-3}, i_{j-3}+1]$.
  The initial part of $P_x$ follows $P_a$, which is disjoint from $P_b$ by Lemma~\ref{lem:segmentdetour};
  the part between $u_{i_j-1}$ and $P$ is disjoint from $P_b$ by Lemma~\ref{lem:segmentdetour}(\ref{item:noothersegments});  and $V_{\sf int}(P)$ is disjoint from $P_b$ by Lemma~\ref{lem:segmentconnections}(\ref{item:pathsp1p2}).
  Let $q$ be the endpoint of $P$ on $P_2$, and let $w$ be the first vertex of $\partner(e_{i_{j-2}})$ on $P_2$. Then $q \leq w$ by Lemma~\ref{lem:segmentconnections}(\ref{item:nicep2order}),
  and we claim $w < V(P_b) \cap V(P_2)$. Note that $w \notin \partner(e_{i_j})$ since the three sets $\partner(e_{i_{j-2}})$, $\partner(e_{i_{j-1}})$, $\partner(e_{i_j})$ are distinct
  and by Observation~\ref{obs:partner_switching}, let $w'$ be the first vertex of $\partner(e_{i_j})$ on $P_2$.
  Assume for a contradiction that $P_b$ intersects $w$. Then $P_b$ provides a path from $w$ to $e_{i_j}$ that avoids $e_{i_{j-1}}$ and $w'$;
  hence $w' \notin \partner(e_{i_{j-1}})$ and $\partner(e_{i_{j-1}})\cap \partner(e_{i_j})=\emptyset$. 
  But since $\partner(e_{i_{j-2}}) \neq \partner(e_{i_{j-1}})$, there must be at least one further vertex $w'' \in \partner(e_{i_{j-1}})$
  such that $w < w'' < w'$; this contradicts that $P_b$ intersects $w$ by Lemma~\ref{lem:segmentconnections}. 
  Thus $P_x$ and $P_b$ are internally vertex-disjoint.
  The argument for $P_y$ is analogous to that for $P_x$. 
  Now $P_a=P_x+(x,y)+P_y$ and $P_b$ form a pair of internally vertex-disjoint $u_{i_j}$-$v_{i_j}$-paths, 
  and since $e_{i_j} \in S$ was chosen arbitrarily, we conclude that $G'-S$ is biconnected.
  Since $G$ is a supergraph of $G'$, $G-S$ is also biconnected.
    
  Finally, it follows from Lemma~\ref{lem:greedy_yes_global_yes} that since $S\subseteq \heavy(\mu(k))$ is a biconnectivity deletion set of size $k$ for $G$, there is a solution for the given instance intersecting $\heavy(\mu(k))$. This completes the proof of the lemma.
\end{proof}

\subsubsection{Identical partner sets}

Due to Lemma~\ref{lem:manypartnersets}, we assume that there are at most $3k$ distinct partner sets for the edges of $\critical_{G'}(e)\cap \heavy(\mu(k))$ which lie on $P_1$. 
Let $e_1, \ldots, e_t$ be the set of all edges of $\critical_{G'}(e)\cap \heavy(\mu(k)) \cap E(P_1)$, in the order they appear on $P_1$ from $x$ to $y$.
We define a set of exceptional edges; initially we set $\hat I=\{ 1 \leq i < t \mid \partner(e_{i}) \neq \partner(e_{i+1})\}$
(later we will define further exceptional edges). Then $|\hat I| \leq 3k$ by Observation~\ref{obs:partner_switching}; 
we study the structure of contiguous stretches of edges $e_i$, \ldots, $e_j$
with indices disjoint from $\hat I$. Note that all edges in such a stretch have identical partner sets. 
We make an observation about the structure.

\begin{lemma}
  \label{lem:samepartners}
  Let $Z=\{e_i, \ldots, e_j\}$ be a set of edges of $\critical_{G'}(e)\cap \heavy(\mu(k))\cap E(P_1)$ such that for every $i \leq i'<j' \leq j$, $e_{i'}$ occurs before $e_{j'}$ when traversing $P_1$ from $x$ to $y$ and $\partner_{G'-e}(e_i)=\partner_{G'-e}(e_j)$. Then the following hold:
  \begin{enumerate}
  \item $|\partner_{G'-e}(e_{i'})|=1$ and $\partner_{G'-e}(e_{i'})$ = $ \partner_{G'-e}(e_i)$ for every $i \leq i' \leq j$, say $\partner_{G'-e}(e_{i'})=\{w\}$;
  \item For every $i \leq i' < j$ and  nice path $P$ for  $\segment[i', i'+1]$,  
      $V(P)\cap V(P_2)=\{w\}$.
  \end{enumerate}
\end{lemma}
\begin{proof}
The first statement follows from Observation~\ref{obs:partner_switching}: since $\partner_{G'-e}(e_i)$ and $\partner_{G'-e}(e_j)$ can intersect in at most one vertex,
we have $|\partner_{G'-e}(e_i)|=1$, and since partner sets appear in an ``ordered'' way on $P_2$, we have $\partner_{G'-e}(e_{i'})=\partner_{G'-e}(e_i)$ for every $i < i' \leq j$.
The second statement follows from the third statement of Lemma~\ref{lem:segmentconnections}. 
%
\end{proof}

\noindent
In light of Lemma \ref{lem:samepartners}, for any edge $e_i$
with $i \notin \hat I$, we let $w(i)$ denote the single partner vertex of $e_i$,
i.e., $\partner_{G'-e}(e_i)=\{w(i)\}$. 

\begin{figure}[t]
  \begin{center}\includegraphics[scale=0.35]{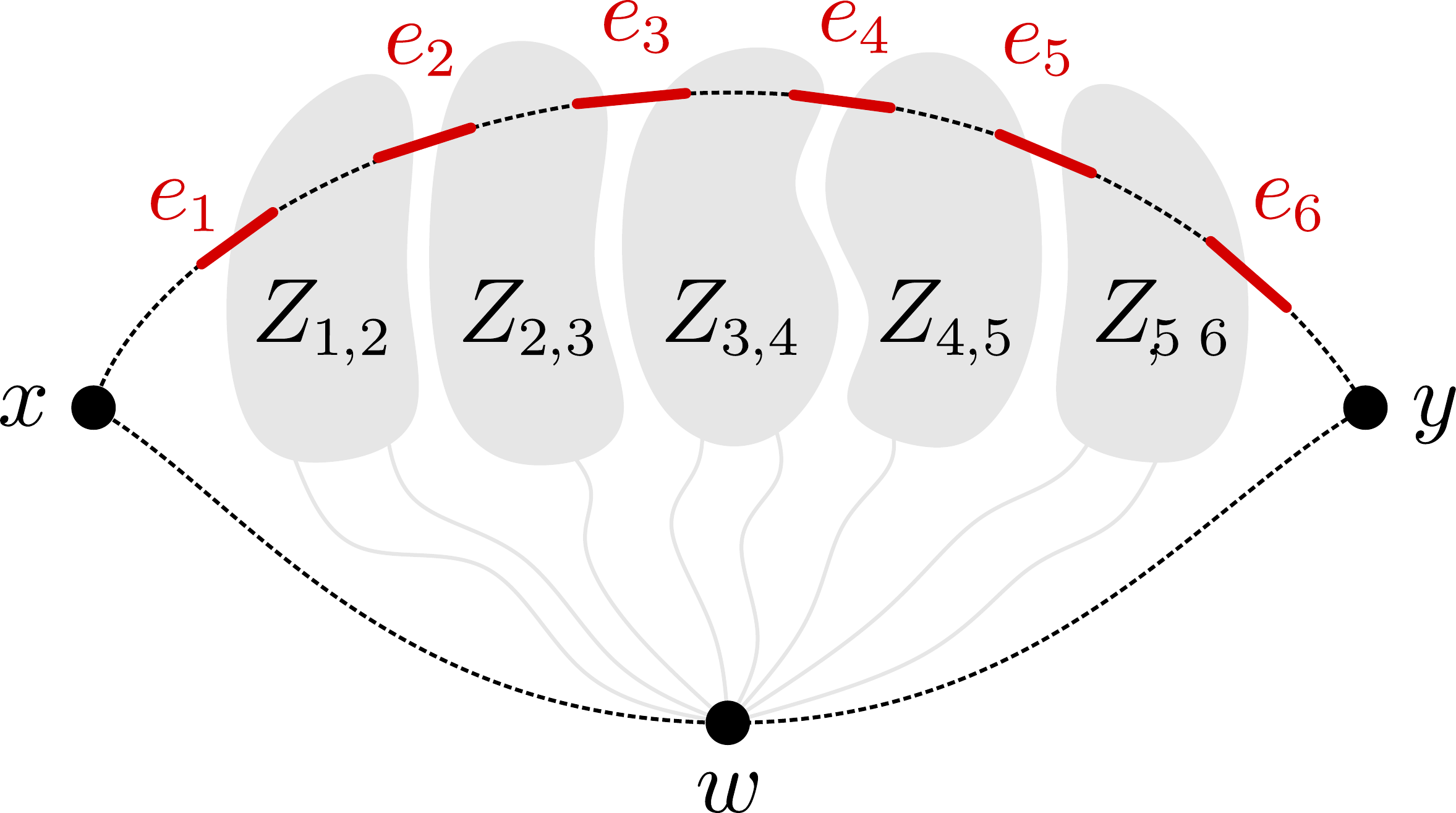}
  	
  \end{center}

  \caption{An illustration of the sets $\{\component[i,i+1]\mid i\in [t-1]\}$ for $t=6$, where the set $Z_{i,i+1}$ represents $\component[i,i+1]$. Note that it is possible that $\component[i,i+1]$ contains a single vertex, in which case this vertex must be the same as both $v_i$ and $u_{i+1}$.}
  \label{fig:components}
\end{figure}

\begin{definition}
  For each $1 \leq i < t$ with $i \notin \hat I$, we define $\component[i,i+1]$ as the set of vertices reachable from  $V(\segment[i,i+1])$ in the graph $G'-\{e_i,e_{i+1},w(i)\}$ (see Figure \ref{fig:components}). We let $\Gamma[i,i+1]$ denote the edge set $E(\component[i,i+1])\cup E(w(i),\component[i,i+1])$. 
\end{definition}

\begin{observation}\label{obs:component_properties}
Let $Z=[t] \setminus \hat I$ be the indices of non-exceptional edges $e_i$. 
The following hold. 
	\begin{itemize}
        \item For every $i \in Z$, $\segment[i,i+1]$ is contained in $\component[i,i+1]$.
        \item For every $i \in Z$, $N_{G'}(\component[i,i+1])=\{u_i,v_{i+1},w(i)\}$.
        \item For every pair $i, j \in Z$, $i \neq j$, 
          the sets $\component[i,i+1]$ and $\component[j,j+1]$ are vertex-disjoint and they are disjoint from $V(P_2)$.
        \item For every pair $i, j \in Z$, $i \neq j$, 
          $\Gamma[i,i+1]$ and $\Gamma[j,j+1]$ are disjoint.
	\end{itemize}
\end{observation}

\begin{proof} 
  The first and second statements follow from the definition of
  $\component[i,i+1]$. For the third  statement, the second statement of
  Lemma~\ref{lem:samepartners} implies that  $\component[i,i+1]$ is disjoint
  from $V(P_2)$ and the first statement of  Lemma~\ref{lem:segmentconnections}
  implies that for every $i\neq j$, the sets $\component[i,i+1]$ and
  $\component[j,j+1]$ are vertex-disjoint. The final statement is a  direct
  consequence of the third statement. This completes the proof. 
\end{proof}

\noindent
We need to consider one further complication. Recall that $\hat F$, 
as defined  after Observation \ref{obs:exists_good_index},
denote the edges removed from the original graph $G$ to create $G'$. 

\begin{definition}
  For each $1\leq i\leq t-1$, $i \notin \hat I$, we say that $\component[i,i+1]$ is  {\em affected} if the set $\hat F$ has an endpoint in $\component[i,i+1]$ and {\em unaffected} otherwise.
\end{definition}

\noindent
Since $|\hat F| < k$, it follows that fewer than $2k$ of these disjoint vertex-sets can be affected. We will treat these as a secondary set of exceptional indices; let $\hat J = \{i \in [t] \setminus \hat I \mid \component[i,i+] \text{ is affected}\}$. 
We make a final observation.

\begin{observation}
  Let $e_1, \ldots, e_t$ be as above, with $t \geq 10k^2+23k$.
  Then there is a contiguous sequence $a, \ldots, b \in [t]$ of indices
  such that $b \geq a + 2k+3$ and for every integral $i \in [a,b]$, $\partner(e_i)=\partner(e_a)$
  and $\component[i,i+1]$ is unaffected. 
\end{observation}
\begin{proof}
  We have $|\hat I|\leq 3k$ and $|\hat J| \leq 2k-2$, hence $[t] \setminus (\hat I \cup \hat J)$
  decomposes into at most $5k-1$ parts.
  With $t \geq (2k+4)(5k-1) + 5k-2 = 10k^2+23k-6$, one of these parts
  will contain at least $2k+4$ indices, hence its bounding indices $a, b$
  will satisfy $b \geq a+2k+3$. 
\end{proof}

\noindent
We refer to such a sequence $e_a, \ldots, e_b$ of edges as an \emph{clean stretch} of $P_1$. 
The remaining task towards the {\sf FPT} algorithm is to show that a sufficiently long
clean stretch contains an irrelevant edge. 

\subsubsection{Reducing clean stretches}

\noindent
We will now restrict our attention to a single clean stretch $[a,b]$, and prove that it contains an irrelevant edge. 
To simplify the notation, let $w=w(a)$. 
We have the following lemma, where $\delta_{H}(Q)$ denotes the edges of $H$ with one endpoint in $Q$.

\begin{lemma} \label{lem:unaffected_component_globally_same}
Let $G'=G-\hat F$ be as above and let $Z=\{e_a, \ldots, e_b\}$ be a clean stretch.
Then for every $a \leq i < b$,
(a) $\component[i,i+1]$ is unaffected, (b) $N_{G'}(\component[i,i+1])$ = $N_{G}(\component[i,i+1])=\{u_i,v_{i+1},w\}$  and (c) $\delta_{G'-w}(\component[i,i+1])=\delta_{G-w}(\component[i,i+1])=\{e_{i},e_{i+1}\}$.
\end{lemma}

\begin{proof} 
Statement (a) holds by definition. 
For statements (b) and (c), the neighbourhoods and incident edges are the same in $G$ as in $G'$ since the components are unaffected, 
and it follows from the definition of $\component[i,i+1]$ that $N_{G'}(\component[i,i+1])=\{u_i,v_{i+1},w\}$ and $\delta_{G'-w}(\component[i,i+1])=\{e_i,e_{i+1}\}$. 
\end{proof}

%
%

\begin{lemma}\label{lem:local_2_flow}
	Let $G',Z=\{e_a,\dots, e_b\}$ be as above. For any $i\in [a+1,b-1]$, the following hold:
	\begin{enumerate}
	\item There is a $v_i$-$\{w,u_{i+1}\}$ flow of value 2 in the graph $G[\component[i,i+1]\cup \{w\}]$.
\item There is a $u_i$-$\{w,v_{i-1}\}$ flow of value 2 in the graph $G[\component[i-1,i]\cup \{w\}]$.
\end{enumerate}
	\end{lemma}

	\begin{proof} We show the first statement; the proof of the second is analogous. If $v_i=u_{i+1}$, then the statement follows by considering the single-vertex path $v_i$ in combination with the nice path for $\segment[i,i+1]$, hence assume that $v_i \neq u_{i+1}$. 
          Since $e_i$ is not critical in $G'$ it follows that there is a $u_i$-$v_i$ flow of value 2 in $G'-e_i$. However, observe that due to Lemma \ref{lem:unaffected_component_globally_same},  $\{w,u_{i+1}\}$ is a $u_i$-$v_i$ separator in $G'-e_i$. Hence of the two paths of the flow, one contains $w$ and the other contains $u_{i+1}$. Truncating these paths at $w$ and at $u_{i+1}$ produces a flow in $G'$.
By Lemma~\ref{lem:unaffected_component_globally_same}, this truncated flow 
must remain in $G'[\component[i,i+1] \cup \{w\}]$, 
and since $G$ is a supergraph of $G'$ it also exists in $G$.
This completes the proof of the statement.
	\end{proof}

%
%

\begin{lemma}\label{lem:spanning_path}
Let $G',Z=\{e_a,\dots, e_b\}$ be as above. For any $i\in [a+1,b-1]$ and $u_i$-$v_i$ path $P$ in $G-w-e_i$, there exists paths $P_1,P_2,P_3$ such that $P=P_1+P_2+P_3$ and the following hold:

\begin{enumerate}
\item $P_1$ is a $u_i$-$u_a$ path such that for every $a\leq j<i$, $P_1$ contains $e_j$
and $v_{j}$ occurs before $u_j$ when traversing $P_1$ from $u_i$ to $u_a$. 

\item $P_3$ is a $v_i$-$v_b$ path such that 
for every $i<j\leq b$, $P_3$ contains $e_j$ and $u_{j}$ occurs before $v_j$ when traversing $P_2$ from $v_i$ to $v_b$. 

\item  $E(P_2)$ disjoint from $\{e_{j}\mid j\in [a,b]\}$ and $V(P_2)$ is disjoint from $\bigcup_{j\in [a,b-1]} \component[j,j+1]$.
	
\end{enumerate}

\end{lemma}

\begin{proof}


Observe that $u_{i}$ lies in the set $\component[i-1,i]$. Furthermore, by Lemma \ref{lem:unaffected_component_globally_same}, we know that $\delta_{G-w-e_i}(\component[i-1,i])=\{e_{i-1}\}$. Since $\component[i-1,i]$ does not contain $v_i$, it must be the case that $P$ contains the edge $e_{i-1}$ and furthermore, $v_{i-1}$ is encountered before $u_{i-1}$ when traversing $P$ from $u_i$ to $v_i$. We can then repeat the same argument for $e_{i-1},e_{i-2}$ and so on until $e_a$. Hence, we conclude that $u_a$ lies on $P$ and that $P_1$, which is the subpath of $P$ from $u_i$ to $u_a$, contains every $e_j$ such that $a\leq j<i$. Furthermore, for every $a\leq j<i$, $v_{j}$ is encountered before $u_j$ when traversing $P_1$ from $u_i$ to $u_a$. This completes the argument for the first statement.

A symmetric argument implies that $v_b$ lies on $P$ and that $P_3$, which is the subpath of $P$ from $v_i$ to $v_b$, contains every $e_j$ such that $i<j\leq b$. Furthermore, for every $i<j\leq b$, $u_{j}$ is encountered before $v_j$ when traversing $P_2$ from $v_i$ to $v_b$.

For the final statement, observe that  $E(P_1)\cup E(P_3)$ contains the set $\{e_a,e_{a+1},\dots, e_{b}\}\setminus e_{i}$. Therefore,  the subpath of $P$ from $u_a$ to $v_b$, which we denote by $P_2$, is disjoint from the set $\{e_a,e_{a+1},\dots, e_{b}\}$. From Lemma \ref{lem:unaffected_component_globally_same}, we infer that the only way $P_2$ can contain a vertex of $\component[j,j+1]$ for some $j\in [a,b-1]$ is if it contains either $e_{j}$ or $e_{j+1}$. Since we have already ruled this out, we conclude that $P_2$ is disjoint from the set $\bigcup_{j\in [a,b-1]} \component[j,j+1]$.  This completes the proof of the lemma.
\end{proof}

\noindent
We are now ready to prove our lemma concerning irrelevant edges.

\begin{figure}[t]
  \begin{center}
    \includegraphics[width=350pt]{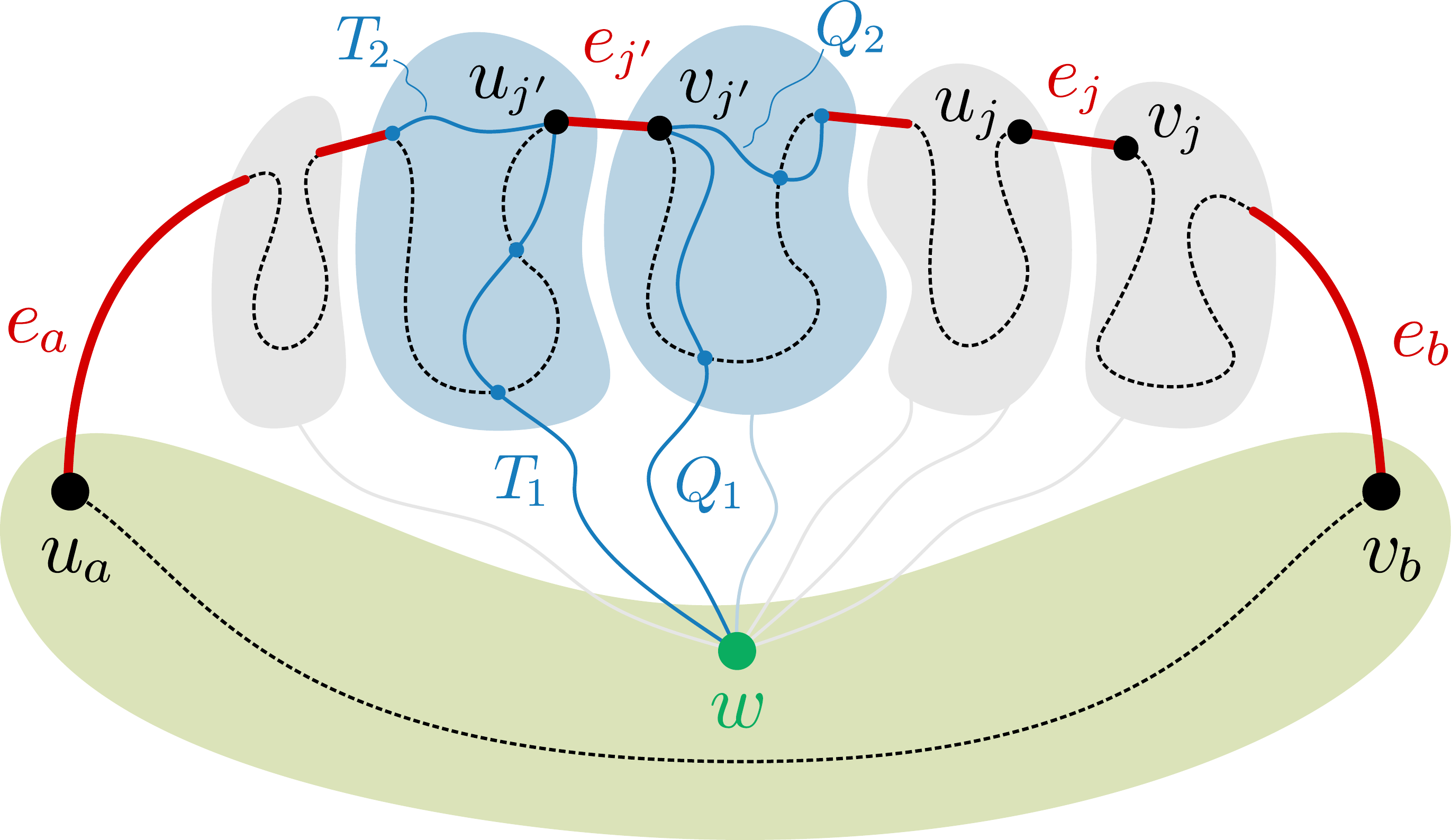}	
  \end{center}

  \caption{An illustration of the paths used in the proof of Lemma \ref{lem:irrelevant_edge}. The dashed lines represent the $u_j$-$v_j$ path $R$.}
  \label{fig:final_paths}
\end{figure}

\begin{lemma}\label{lem:irrelevant_edge}
	Let $G',Z=\{e_a,\dots, e_b\}$ be as above where $b \geq a + 2k+3$.
        Let $j\in [a+1,b-1]$ be such that for every  $f\in \{e_{a+1},\dots, e_{b-1}\}$, $w(f)\geq w(e_j)$. Then, $e_j$ is irrelevant.
\end{lemma} 

\begin{proof}
In order to prove the lemma, we need to argue that if there is a solution for the instance $(G,k,w^*,w,E^\infty)$, then there is one which does not contain $e_j$. Let $S$ be a solution for this instance. If $S$ is disjoint from $e_j$, then we are done. Suppose that this is not the case. We first argue that there is an edge which is not in $S$ and has certain special properties. We will then argue that replacing $e_j$ with this special edge also leads to a solution for the same instance.

\begin{claim} 
	There exists $j'\in [a+1,b-1]$ such that $S$ is disjoint from $\Gamma[j'-1,j']\cup\{e_{j'}\}\cup  \Gamma[j',j'+1]$.
\end{claim}

\begin{proof}
  We first observe that $S$ is disjoint from $\{e_a, \ldots, e_b\} \setminus \{e_j\}$
  by Lemma~\ref{lem:spanning_path}. 
  Next, for $\ell\in \{0,1,\dots ,k\}$ let $\ell'=a+2\ell$. Since there are $2k+2$ edges in the set $\{e_{a+1},\dots, e_{b-1}\}$ and by Observation \ref{obs:component_properties}, it follows that there are $k+1$ edge-disjoint sets $K_1,\dots, K_{k+1}$   where  
  \[
    K_\ell = \Gamma[\ell'+1,\ell'+2]  \quad \bigcup \quad \Gamma[\ell'+2,\ell'+3].
  \]
  Since $|S|\leq k$, we conclude that there is an index $j'$ such that $S$ is disjoint from $\Gamma[j'-1,j']\cup\{e_{j'}\}\cup  \Gamma[j',j'+1]$. This completes the proof of the claim. 
\end{proof}

\noindent
Furthermore, by our choice of $e_j$, it follows that $w(e_{j'})\geq w(e_j)$. Let $S'=S\setminus \{e_j\}\cup \{e_{j'}\}$. Clearly, $|S'|\leq k$ and $w(S')\geq w(S)$. We now argue that $S'$ is also a biconnectivity deletion set.\looseness-1

Observe that in order to do so, it suffices to prove that there is a $u_{j'}$-$v_{j'}$ flow of value 2 in $G-S'$. Let ${\cal Q}=\{Q_1,Q_2\}$ 
be a  $v_{j'}$-$\{w,u_{j'+1}\}$ flow of value 2 in the graph $G[\component[j',j'+1]\cup \{w\}]$ where $Q_1$ is incident with $w$ and $Q_2$ with $u_{j'+1}$ (see Figure \ref{fig:final_paths}). This flow exists by Lemma \ref{lem:local_2_flow}.
 Similarly, let ${\cal T}=\{T_1,T_2\}$ be a $u_{j'}$-$\{w,v_{j'-1}\}$ flow of value 2 in the graph $G[\component[j'-1,j']\cup \{w\}]$ where $T_1$ is incident with $w$ and $T_2$ with $v_{j'-1}$. By Observation \ref{obs:component_properties}, we know that $\component[j',j'+1]$ and $\component[j'-1,j']$ are vertex-disjoint. Therefore, $J_1=Q_1+T_1$ is a $u_{j'}$-$v_{j'}$ path in $G-S'$ and furthermore, $J_1$ is contained in $G[\component[j'-1,j']\cup \component[j',j'+1]\cup \{w\}]$. 
 
Since $S$ is a solution containing $e_j$, it follows that there is a $u_{j}$-$v_{j}$ flow of value 2 in the graph $G-S$. In particular, there is a $u_j$-$v_j$ path $R$ in the graph $G-S-w-e_j$. 
We assume without loss of generality that $j>j'$. The arguments in the other case are exactly the same. Due to Lemma \ref{lem:spanning_path}, we know that $R=R_1+R_2+R_3$ where the paths $R_1$, $R_2$ and $R_3$ satisfy the stated properties.

Let $R_1'$ be the subpath of $R_1$ from $u_j$ to $u_{j'+1}$, $R_1''$ be the subpath of $R_1$ from  $v_{j'-1}$ to $u_{a}$. Now, observe that $H=R_1'+Q_2+e_{j'}+T_2+R_1''+R_2+R_3$ is also a $u_j$-$v_j$ path in $G-S-w-e_j$. Futhermore, $H$ intersects $J_1$ only in $\{u_{j'},v_{j'}\}$. 

Now, $C=H+e_j$ is a cycle in $G-S-w$ such that  $S\cap E(C)=e_j$ and $S'\cap E(C)=e_{j'}$. 
Therefore, $J_2=C-e_{j'}$ is a $u_{j'}$-$v_{j'}$ path in $G-S-w$. Since $C$ intersects $J_1$ only in $\{u_{j'},v_{j'}\}$, we conclude that $J_2$ is internally vertex-disjoint from $J_1$ and contains no edges of $S'$. Thus, we have demonstrated a $u_{j'}$-$v_{j'}$ flow $\{J_1,J_2\}$ of value 2 in the graph $G-S'$, implying that $S'$ is a biconnectivity deletion set and hence a solution for the instance $(G,k,w,w^*,E^\infty)$. Therefore, we conclude that the edge $e_j$ is irrelevant, completing the proof of the lemma.
\end{proof}

\noindent
Combining Lemma~\ref{lem:manypartnersets} and the fact that we can clearly
locate the irrelevant edge $e_j$ (in the statement of
Lemma~\ref{lem:irrelevant_edge}) in polynomial time, we obtain
Lemma~\ref{lem:biconnected_bound_on_non_critical_edges}, our main
objective. This concludes the description of our algorithm for {\twoconndel}.

\subsection{A randomized kernel for {\sc Unweighted Biconnectivity Deletion}}
\label{subsubsec:kernel}

We now present our randomized kernel for the {\twoconndel} problem where instances are of the form $(G,k,w^*,w,E^\infty)$ where $w(e)=1$ for every $e\in E(G)\setminus E^\infty$, $w(e)=0$ for every $e\in E^\infty$, and $w^*=k$. This version of the problem will be referred to as {\sc Unweighted Biconnectivity Deletion} and instances of this problem will henceforth be of the form $(G,k,E^\infty)$ where a solution is a biconnectivity deletion set of size $k$ contained in $E(G)\setminus E^\infty$. We continue to refer to the set $E(G)\setminus E^\infty$ as the set of potential solution edges and assume without loss of generality that at any point, any edge in the set $\critical_G(\emptyset)$ is already part of $E^\infty$.
Finally, recall that a \emph{linkage} from $A$ to $B$ in a digraph $D$, where $A$
and $B$ are vertex sets, is a collection of $|A|=|B|$ pairwise
vertex-disjoint paths originating in $A$ and terminating in $B$.

Our kernelization relies on a result of Kratsch and Wahlstr\"om~\cite{KratschW12}. Before we are able to state it formally, we need the following  definitions.  Let us define a \emph{potentially overlapping $A$-$B$ vertex cut} in a
digraph $D$ to be a set of vertices $C \subseteq V(D)$ such that $D-C$
contains no directed path from $A \setminus C$ to $B \setminus C$.
For any digraph $D$ and set $X \subseteq V(D)$, a set $Z\subseteq V(D)$
   is called a  \emph{cut-covering set} for $(D,X)$ if 
  for any $A, B, R \subseteq X$, 
  there is a minimum-cardinality potentially overlapping 
  $A$-$B$ vertex cut $C$ in $D-R$
  such that $C \subseteq Z$. We are now ready to state the result of Kratsch and Wahlstr\"om on which our kernelization is based.

\begin{lemma}[Corollary 3,  \cite{KratschW12}]
  \label{lemma:cutcover}
  Let $D$ be a directed graph 
  and let $X \subseteq V(D)$.
  We can identify a cut-covering set $Z$  for $(D,X)$ of size $O(|X|^3)$ in polynomial time with failure probability $O(2^{-|V(D)|})$.
%
\end{lemma}

\noindent
Armed with this lemma, we first give a randomized kernelization that outputs
an instance whose size is bounded polynomially in the number of the potential
solution edges in the input instance.

\begin{lemma} \label{lemma:kernel-representation}
 
  \textsc{Unweighted Biconnectivity Deletion} has a randomized kernel  
with number of vertices bounded by $O(|E(G)\setminus E^\infty|^3)$.
\end{lemma}
\begin{proof}
Let $F=E(G)\setminus E^\infty$ be the set of potential solution edges. 	
  Now, the kernelization task essentially consists of retaining enough
  information from the input graph $G$ to verify for any set $S \subseteq F$,  
  whether $S$ is a biconnectivity deletion set for $G$. 
  Observe that this is equivalent to verifying whether there exists an edge $e=(u,v) \in S$, such that the maximum value of a $u$-$v$ flow in $G-S$ is less than 2. 
  We show an equivalent formulation of this as a question about the existence of 
  linkages in an auxiliary digraph, followed by an appropriate invocation of Lemma~\ref{lemma:cutcover}.

  For the formulation, we create a digraph $D_{G,F}$ from $G$ and $F$. We refer to this digraph as $D$ when $G$ and $F$ are clear from the context. In the first step, subdivide every edge $e \in F$ with a new vertex $x_e$. That is, for an edge $e=(u,v)\in F$, we create a new vertex $x_e$, remove the edge $e$ and add edges $(u,x_e)$ and $(v,x_e)$. 
  Let $G_1$ be the resulting undirected graph. In the second step, replace every edge $(u,v)$ in $E(G_1)$ by a pair of arcs $(u,v)$, $(v,u)$. 
  Finally, for every vertex $v$ incident to any edge of $F$ in $G$,  
  add vertices $v^+,v^-$ and add arcs from $v^+$ to all vertices in 
  $N_{G_1}(v)$ and from all vertices in $N_{G_1}(v)$ to $v^-$. Let $D$ be the resulting digraph. Note that $N_D^+(v^-)=\emptyset$ and 
 $N_D^-(v^+)=\emptyset$. 
%
  Let $X_E=\{x_e \mid e \in F\}$,   $X_V=\{v^+, v^-, v \mid e \in F, e=(u,v)\}$ and $X=X_E\cup X_V$. 
We now relate solutions for the given instance and linkages in $D$.

  \begin{claim}\label{clm:linkage_equivalence}
    For any $S \subseteq F$, $S$ is a biconnectivity deletion set for $G$ 
    if and only if for every edge $(u,v) \in S$
    there is a linkage from $\{u^+, u\}$ to $\{v^-, v\}$ in $D-\{x_e \mid e \in S\}$.
  \end{claim}
  \begin{proof}
    Consider an arbitrary edge $e=(u,v)$ in $S$. 
    On the one hand, assume that there exists a $u$-$v$ flow of value 2 in $G-S$. 
    Then, by definition there exists a pair $\{P_1,P_2\}$ of
    internally vertex-disjoint $u$-$v$-paths in $G-S$. Observe that orienting both paths 
    from $u$ to $v$, replacing one copy of $u$ by $u^+$ and one copy of 
    $v$ by $v^-$, and subdividing any edge $e' \in E(P_i) \cap F$ by the vertex $x_{e'}$
    yields the required linkage. 

    On the other hand, let $\{P_1, P_2\}$ be a linkage from $\{u^+,u\}$
    to $\{v^-, v\}$ in $D-\{x_e \mid e \in S\}$. For each $i\in \{1,2\}$, if $P_i$ originates in $u^+$, then replace $u^+$ by $u$ and if $P_i$ terminates in $v^-$, then replace $v^-$ by $v$. Call the paths resulting from $P_1$ and $P_2$ in this way,  $P_1'$ and $P_2'$ respectively.
     Then $P_1'$ and $P_2'$ use only vertices 
    of $V(G) \cup X_E$. Furthermore, these two paths use no edge of $S$
    since by definition, $P_1$ and $P_2$ are disjoint from every vertex $x_{e'}$ such that $e' \in S$. 
    Thus the paths $P_1'$, $P_2'$ use only edges and vertices
    present in $G-S$, and form an internally vertex-disjoint pair of $u$-$v$-paths. 

    Since the above applies to any edge,
    the claim follows.
  \end{proof}
  
  \noindent
  Let $Z \subseteq V(D)$ be the cut-covering set for $(D, X)$,
  as computed by the algorithm of Lemma~\ref{lemma:cutcover}. %
%
  Having in hand the set $Z$, we define the  set $Y=(Z\cap V(G)) \cup V(F)$.  
   Note that $Z$ could contain vertices from $X_V$, but we want $Y$ to be a subset of $V(G)$. Therefore, we first add to $Y$ those vertices in $Z$ which are also vertices in $G$ and then add the vertices of $V(F)$. Our objective now is to  reduce $G$ down to what is commonly known as the \emph{torso} graph of $G$ defined by $Y$ (see~\cite{KratschW12}).
  We now make this precise in the form of reduction rules. 
  In the rest of the proof of the lemma, we fix $Z$ to be a set computed using  Lemma~\ref{lemma:cutcover} and let $Y$ be as defined above. We now state three reduction rules which will be applied on the given instance in the order in which they are presented.

  \begin{redr}\label{redrule:zerothrule} 
    If $k=0$, then return an arbitrary yes-instance of constant size.
  \end{redr}
  
  \begin{redr}\label{redrule:firstrule} 
    Suppose that Reduction Rule \ref{redrule:zerothrule} has been applied on
    the given instance. If there is an edge $(u,v) \in F$ such that $G$
    contains a $u$-$v$ path avoiding all edges of $F$ and all vertices of $Y
    \setminus \{u,v\}$, then delete $(u,v)$ from $G$ and reduce the budget $k$
    by 1. That is,  return the instance $(G-\{(u,v)\}, k-1,E^\infty)$.
  \end{redr}

  \begin{redr}\label{redrule:secondrule}
    Suppose that Reduction Rule \ref{redrule:zerothrule} and Reduction Rule
    \ref{redrule:firstrule} have been applied exhaustively on the given
    instance. For every pair $u,v\in Y$ such that $(u,v)\notin E(G)$ and there
    is a $u$-$v$-path in $G$ that is internally vertex-disjoint from $Y$, we
    add the edge $(u,v)$. Finally, return the instance $(G',k,E'^\infty)$,
    where $G'=G[Y]$ and $E'^\infty=(E^\infty\cap E(G'))\cup (E(G')\setminus
    E(G))$.
  \end{redr}


\noindent
The soundness of Rule~\ref{redrule:zerothrule} is trivial and we move on 
to prove the soundness of the remaining two rules.

  \begin{claim}
Reduction Rules~\ref{redrule:firstrule} and~\ref{redrule:secondrule} are sound.
    \end{claim}
    

  \begin{proof}
  
   Let $e=(p,q)\in F$ be an edge which is deleted in an application of Reduction Rule \ref{redrule:firstrule}. Observe that in order to argue the soundness of this reduction rule, it suffices to argue that $e$ is part of some  solution for the given instance (if there exist any). Let $S$ be an arbitrary subset of $F$ containing $e$ such that $S\setminus \{e\}$ is a solution. If $S$ itself is a biconnectivity deletion set then we may correctly conclude that $e$ is part of some  solution for the given instance. Suppose that this is not the case.

  Recall that by the previous claim, $S$ is a biconnectivity deletion set for $G$
  if and only if there is a linkage from $\{u^+, u\}$ to $\{v^-, v\}$
  in $D-\{x_e \mid e \in S\}$ for every $(u,v) \in S$. Since we are in the case that $S$ is {\em not} a biconnectivity deletion set,  there is a $(u,v)\in S$, with $A=\{u^+, u\}$, $B=\{v^-, v\}$, and $R=\{x_e \mid e \in S\}$  such that there is no linkage from $A$ to $B$
  in $D-R$. Since $S\setminus \{e\}$ is a biconnectivity deletion set, we may assume without loss of generality that $u=p$ and $v=q$ and furthermore, $\kappa_{G-S}(p,q)=1$. In addition, the fact that $Z$ is a cut-covering set for $(D,X)$ implies that $Z$ contains a vertex $w$ such that $C=\{w\}$ is a minimum-cardinality potentially overlapping 
  $A$-$B$ vertex cut in $D-R$. It is straightforward to see that $w\notin \{p,q,p^+,q^-\}$ since otherwise, there will be at least one path from $A$ to $B$ which is disjoint from $w$.  
  Finally, since $\kappa_{G-S}(p,q)=1$, it follows that every $p$-$q$ path in $G-S$ intersects $w$. If $w\in X_E$ then we know that it  corresponds to an edge in $F$. Otherwise, it corresponds to a vertex in $Y$. In either case, we obtain a contradiction to the applicability of Reduction Rule \ref{redrule:firstrule} on the edge $(p,q)$, completing the proof of soundness for this rule.   

We now argue the soundness of Reduction Rule~\ref{redrule:secondrule}.
To do so, we prove that $S\subseteq F$ is a solution for $(G,k,E^\infty)$ if and only if it is a solution for $(G',k,E'^\infty)$. Let $D_1=D_{G,F}$ and let $D_2=D_{G',F}$. 

In the forward direction, suppose that $S$ is a solution for $(G,k,E^\infty)$. By Claim, \ref{clm:linkage_equivalence}, it follows that for every edge $(u,v)\in S$, there is a linkage from $\{u^+,u\}$ to $\{v^-,v\}$ in $D_1-\{x_e\mid e\in S\}$. Fix such an edge $(u,v)$ and  let the paths in the linkage be $P_1,P_2$.  If we demonstrate such a linkage in $D_2$, then we are done. This can be achieved as follows. Let  $i\in \{1,2\}$ and consider a pair of vertices $x_i,y_i\in V(P_i)\cap Y$ such that the subpath of $P_i$ from $x_i$ to $y_i$ has all its internal vertices disjoint from $Y$. Then, we know that the graph $G'$ contains the edge $(x_i,y_i)$ and hence the digraph $D_2$ contains the arc $(x_i,y_i)$. We replace the subpath from $x_i$ to $y_i$ with the arc $(x_i,y_i)$ and we do this for every such subpath of $P_i$. It is straightforward to see that what results is indeed a linkage from $\{u^+,u\}$ to $\{v^-,v\}$ in $D_2-\{x_e\mid e\in S\}$. Hence, we conclude that $S$ is a solution for $(G',k,E'^\infty)$. 

 The same argument can be reversed for the converse direction in order to convert, for any $(u,v)\in S$,  a linkage from $\{u^+,u\}$ to $\{v^-,v\}$ in $D_2-\{x_e\mid e\in S\}$ to a a linkage from $\{u^+,u\}$ to $\{v^-,v\}$ in $D_1-\{x_e\mid e\in S\}$. This completes the proof of soundness of Reduction Rule \ref{redrule:secondrule}. 
  \end{proof}
  
  \noindent
  The above claim implies that if $(G',k',E(G')\setminus F')$ is the instance obtained by exhaustively applying the three reduction rules above, then  $(G',k',E(G')\setminus F')$ is indeed equivalent to  $(G,k,E^\infty)$. Furthermore, the size $|V(G')|=O(|F|^3)$ and the randomized polynomial running time
  follow from Lemma~\ref{lemma:cutcover}. This completes the proof of the lemma.
\end{proof}

\kernelbiconnected*

\begin{proof}
Let $(G,k,E^\infty)$ be the given instance and let $F=E(G)\setminus E^\infty$ be the set of potential solution edges in this instance. We present reduction rules which   reduce $F$ (while maintaining equivalence) to size $O(k^3)$; the result then follows
  from Lemma~\ref{lemma:kernel-representation}. 

  If $|F| = O(k^3)$, we are done. Otherwise, following the approach described in Section \ref{subsubsec:greedy}, we 
  greedily construct a biconnectivity deletion set in $G$, at each step keeping track of the
  edges that become critical. That is, we let $\hat S=\{f_1,\dots, f_r\}\subseteq F$ be a set greedily constructed as follows. The edge $f_1$ is an arbitrary edge in $F$ and for each $2\leq i\leq r$, $f_i$ is an arbitrary edge which is \emph{not} critical in $G-\{f_1,\dots, f_{i-1}\}$. As earlier, we terminate this procedure after $k$ steps if we manage to find edges $\{f_1,\dots, f_k\}$ or earlier if for some $r<k$, every remaining edge of $F$  is critical in $G-\{f_1,\dots, f_r\}$. 
  
  If $r=k$, then we identify the instance as a yes-instance and return an arbitrary yes-instance of constant size. Otherwise, if there is an $i\in [r]$ such that $G-\{f_1,\dots, f_{i}\}$ is biconnected and  $|\critical_{G- \{f_1,\dots, f_{i-1}\}}(f_i)|\geq 20k^2+46k$, then we execute the case analysis in
 Section~\ref{subsubsec:case_analysis} and in polynomial time, either find $3k+1$ distinct partner sets or an irrelevant
  edge. In the latter case, 
  we simply remove this irrelevant edge from $F$ (add it to the set $E^\infty$).
  Finally, if we reach a case with at least $3k+1$ distinct partner sets,
  then according to the proof of Lemma~\ref{lem:manypartnersets} we can find a biconnectivity deletion set
  $S \subseteq F$ with $|S| \geq k$ in polynomial time, and since we are dealing with the unweighted case, we can simply
  identify the instance as a yes-instance and return an arbitrary yes-instance of constant size.

  The only remaining case is that this  greedy algorithm fails to produce
  a large enough solution yet never marks too many edges as critical at once. That is, it terminates in   $r<k$ steps and never marks more than $20k^2+46k$ edges as critical in step $i$ for any $i\in [r]$.  
  This implies that $|F| \leq 20k^3+46k^2+k=O(k^3)$, completing the proof of the theorem.
\end{proof}

\section{Conclusions}
Our results on {\strongcontract} and {\twoconndel} provide additional
data points for the algorithmic landscape of graph editing problems under
connectivity constraints and its application in network design. 

Since we established that {\strongcontract}
is {\sf W[1]}-hard for general digraphs, we ask whether the problem becomes
{\sf FPT} when restricted to planar digraphs or other structurally sparse
classes.

Concerning the parameterized algorithm for {\twoconndel}, we ask whether
the dependence of~$2^{O(k \log k)}$ can be improved to single-exponential
or proven to be optimal. Naturally, we would further like to know whether
we can reach beyond \emph{bi}connectivity and extend our algorithm to 
higher values of vertex-connectivity. Is it  possible to obtain a similar algorithm on digraphs?

Finally, regarding our polynomial kernel for {\sc Unweighted Biconnectivity Deletion},
we ask whether it is possible to obtain a deterministic kernel. It is also
left open whether the weighted case admits a polynomial kernel.

The results presented in this paper raise more questions than they answer,
a clear indication that connectivity constraints are far from properly
explored under the paradigm of parameterized complexity. As such, the
topic offers exciting but challenging opportunities for further research.

\end{document}